\documentclass[smallextended]{svjour3}       
\smartqed  
\usepackage{graphicx}

\usepackage{latexsym}
\usepackage{amssymb,amsbsy,amsmath,amsfonts,amssymb,amscd,amsfonts,mathrsfs}
\usepackage{graphicx}
\usepackage{color}

\newcommand {\sgn} { {\rm sgn} }
\newcommand {\rmd} { {\mathrm d} }

\newcommand{\beq}{\begin{equation}}
\newcommand{\eeq}{\end{equation}}
\newcommand{\bea} {\begin{array}{rl}}
\newcommand{\eea} {\end{array}}
\newcommand{\bepa}{\left\{ \begin{array}{l}}
\newcommand{\eepa} {\end{array}\right.}

\begin{document}

\title{Bridging the gap between individual-based and continuum models of growing cell populations
}

\titlerunning{From individual-based to continuum models of growing cell populations}        

\author{Mark A J Chaplain \and Tommaso Lorenzi \and Fiona R Macfarlane 
}


\institute{
	   Mark A J Chaplain \at
            School of Mathematics and Statistics, University of St Andrews, St Andrews, KY16 9SS
            \email{majc@st-andrews.ac.uk} 
	    \and	
	    Tommaso Lorenzi \at
            School of Mathematics and Statistics, University of St Andrews, St Andrews, KY16 9SS
            \email{tl47@st-andrews.ac.uk} 
	     \and
             Fiona R Macfarlane \at
             School of Mathematics and Statistics, University of St Andrews, St Andrews, KY16 9SS\\
             \email{frm3@st-andrews.ac.uk}           
}

\date{Received: date / Accepted: date}

\maketitle

\begin{abstract}
Continuum models for the spatial dynamics of growing cell populations have been widely used to investigate the mechanisms underpinning tissue development and tumour invasion. These models consist of nonlinear partial differential equations that describe the evolution of cellular densities in response to pressure gradients generated by population growth. Little prior work has explored the relation between such continuum models and related single-cell-based models. We present here a simple stochastic individual-based model for the spatial dynamics of multicellular systems whereby cells undergo pressure-driven movement and pressure-dependent proliferation. We show that nonlinear partial differential equations commonly used to model the spatial dynamics of growing cell populations can be formally derived from the branching random walk that underlies our discrete model. Moreover, we carry out a systematic comparison between the individual-based model and its continuum counterparts, both in the case of one single cell population and in the case of multiple cell populations with different biophysical properties. The outcomes of our comparative study demonstrate that the results of computational simulations of the individual-based model faithfully mirror the qualitative and quantitative properties of the solutions to the corresponding nonlinear partial differential equations. Ultimately, these results illustrate how the simple rules governing the dynamics of single cells in our individual-based model can lead to the emergence of complex spatial patterns of population growth observed in continuum models.

\end{abstract}

\section{Introduction}
\label{intro}
\paragraph{Brief overview of continuum models of growing cell populations.} Continuum models for the spatial dynamics of growing cell populations have been widely used to complement empirical research in developmental biology and cancer research. These models consist of nonlinear partial differential equations that describe the evolution of cellular densities in response to pressure gradients generated by population growth, which can be mechanically regulated, nutrient-limited or pressure-dependent~\cite{ambrosi2002mechanics,ambrosi2002closure,araujo2004history,bresch2010computational,byrne2010dissecting,byrne1995growth,byrne1996growth,byrne1997free,byrne2009individual,byrne2003two,byrne2003modelling,chaplain2006mathematical,chen2001influence,ciarletta2011radial,greenspan1976growth,lowengrub2009nonlinear,perthame2014some,preziosi2003cancer,ranft2010fluidization,roose2007mathematical,sherratt2001new,ward1999mathematical,ward1997mathematical}.

For a population of cells with growth rate that depends on the local pressure~\cite{bru2003universal,byrne2003modelling,drasdo2012modeling,ranft2010fluidization}, a prototypical example of such models is given by the following equation 
\beq
\label{eq:pde1}
\partial_t \rho - \mu \, {\rm div}\left(\rho \nabla p \right) = G(p) \rho,
\eeq
which was proposed by Byrne \& Drasdo \cite{byrne2009individual}. Equation~\eqref{eq:pde1} is a conservation equation for the function $\rho(t,x) \geq 0$ that represents the density of cells at position $x\in\mathbb{R}^{d}$, with $d=1,2,3$ in the biologically relevant cases, and time $t \in \mathbb{R}^+$. The function $p$ stands for the cell pressure and the term $G$ is the net growth rate of the cell density. 

The second term on the left-hand side of equation~\eqref{eq:pde1} represents the rate of change of the cellular density due to pressure-driven cell movement (\emph{i.e.}  the tendency of cells to move towards regions where they feel less compressed). The definition of this term builds on the seminal paper of Greenspan~\cite{greenspan1976growth} and subsequent work of Byrne \& Chaplain~\cite{byrne1997free}. In analogy with the classical Darcy's law for fluids, the parameter $\mu > 0$ can be seen as the cell mobility coefficient, which is defined as the quotient between the permeability of the porous medium in which the cells are embedded (\emph{e.g.} the extracellular matrix) and the cellular viscosity. 

The term on the right-hand side of equation~\eqref{eq:pde1} represents the rate of change of the cellular density due to cell proliferation (\emph{i.e.} cell division and cell death). In the mathematical framework of equation~\eqref{eq:pde1}, the effect of pressure-limited growth can be captured in a first approximation by assuming the net growth rate $G$ to be a smooth function of $p$ that satisfies the following assumptions
\beq
\frac{\rmd G}{\rmd p} < 0, \quad G(P) =0. 
\label{as:G}
\eeq
In assumptions~\eqref{as:G}, the parameter $P > 0$ stands for the pressure at which cell death exactly compensates cell division. The term homeostatic pressure has been coined to indicate such a critical pressure~\cite{basan2009homeostatic}. 

In order to close equation~\eqref{eq:pde1} the pressure $p$ can be defined according to a barotropic relation $p \equiv \Pi(\rho)$. Typically, the function $\Pi(\rho)$ is identically zero for $\rho~\leq~ \rho^*$ and is monotonically increasing for $\rho > \rho^*$, with $0<\rho^*<\Pi^{-1}(P)$ being the density below which cells do not exert any force upon one another~\cite{tang2014composite}. For simplicity, in this paper we let $\Pi$ be a smooth function of $\rho$ that satisfies the following assumptions 
\beq
\Pi(0) = 0,  \quad \frac{\rmd \Pi}{\rmd \rho}>0 \; \text{ for } \; \rho>0.
\label{as:p}
\eeq
For instance, in the attempt to reduce the biological problem to its essentials while ensuring analytical tractability of the mathematical model, Perthame {\it et al.}~\cite{perthame2014hele} have proposed the following definition of $\Pi(\rho)$, which satisfies assumptions~\eqref{as:p}:
\beq
\label{def:p}
\Pi(\rho) = K_{\gamma} \, \rho^{\gamma} \quad \text{with} \quad \gamma > 1 \quad \text{and} \quad K_{\gamma}>0.
\eeq
In definition~\eqref{def:p}, the parameter $\gamma$ provides a measure of the stiffness of the barotropic relation and $K_{\gamma}$ is a scale factor. The limit $\gamma \to \infty$ corresponds to the case where cells behave like an incompressible fluid. In this asymptotic regime, it has been proven that models of the form of equation~\eqref{eq:pde1} converge to free-boundary problems of Hele-Shaw type~\cite{kim2018porous,kim2016free,mellet2017hele,perthame2014derivation}.

The model~\eqref{eq:pde1} can be generalised to the case of multiple cell populations with different biophysical properties (\emph{i.e.} different mobilities and growth rates) as follows
\beq
\label{eq:pde2}
\partial_t \rho_h - \mu_h \, {\rm div}\left(\rho_h \nabla p \right) = G_h(p) \rho_h, \quad h=1,\ldots,M.
\eeq
The system of coupled equations~\eqref{eq:pde2} relies on notation and assumptions analogous to those underlying equation~\eqref{eq:pde1}. In particular, the coefficient $\mu_h > 0$ represents the mobility of cells in the $h^{th}$ population and the pressure $p$ is defined by a barotropic relation $p \equiv \Pi(\rho)$. Here $\rho$ stands for the total cell density, \emph{i.e.} 
\beq
\label{def:rho}
\rho(t,x) = \sum_{h=1}^M \rho_h(t,x),
\eeq
and the function $\Pi$ satisfies conditions~\eqref{as:p}. Moreover, the net growth rate of the cell density $G_h(p)$ can be assumed to be a smooth function of the cell pressure that satisfies assumptions~\eqref{as:G} for all $h~=~1,\dots,M$. 

For example, building on the computational results presented by Drasdo \& Hoehme~\cite{drasdo2012modeling}, Lorenzi {\it et al.}~\cite{lorenzi2017interfaces} considered the following variant of the system of equations~\eqref{eq:pde2}
\beq
\label{eq:pde3}
\left\{
\begin{array}{ll}
\partial_t \rho_1 - \mu_1 \, {\rm div}\left(\rho_1 \nabla p \right) = G(p) \rho_1, 
\\\\
\partial_t \rho_2 - \mu_2 \, {\rm div}\left(\rho_2 \nabla p \right) = 0,
\end{array}
\right.
\eeq
complemented with the barotropic relation~\eqref{def:p}. The system of equations~\eqref{eq:pde3} describes the interaction dynamics between a population of proliferating cells (\emph{i.e.} population $1$) and a population of nonproliferating cells (\emph{i.e.} population $2$) with different mobilities.

Further to the biological and clinical insights into the underpinnings of tissue development and tumour growth they can provide, these continuum models exhibit a range of interesting qualitative behaviours. For instance, it was shown that models in the form of equation~\eqref{eq:pde1} admit travelling-wave solutions with composite shapes and discontinuities~\cite{tang2014composite}. Moreover, in analogy with reaction-diffusion systems arising in the mathematical modelling of other biological and ecological problems~\cite{dancer1999spatial,mimura2000reaction}, models like the system of equations~\eqref{eq:pde3} can give rise to sharp interfaces, which bring about spatial segregation between cell populations with different biophysical properties~\cite{lorenzi2017interfaces}. 

\paragraph{Derivation of continuum models of growing cell populations from individual-based models.} A key advantage of continuum models for the spatial dynamics of growing cell populations over their individual-based counterparts (\emph{i.e.} discrete models that track the dynamics of individual cells)~\cite{drasdo2005coarse,van2015simulating} is that they are amenable to mathematical analysis. This enables a complete exploration of the model parameter space, which ultimately allows more robust conclusions to be drawn. Moreover, compared to individual-based models, continuum models offer the possibility to carry out numerical simulations at the level of larger portions of tissues or even of whole organs, while keeping computational costs within acceptable bounds. 

Since continuum models are defined at the scale of whole cell populations, they are usually formulated on the basis of phenomenological considerations, which can hinder a precise mathematical description of crucial biological and physical aspects. On the contrary, stochastic individual-based models that describe the dynamics of single cells in terms of algorithmic rules can be more easily tailored to capture fine details of cellular dynamics, thus making it possible to achieve a more accurate mathematical representation of multicellular systems. Furthermore, individual-based models are able to reproduce the emergence of population-level phenomena that are induced by stochastic fluctuations in single-cell biophysical properties, which are relevant in the regime of low cellular densities and cannot easily be captured by continuum models. Therefore, it is desirable to derive continuum models for the spatial dynamics of cell populations as the appropriate limit of individual-based models for spatial cell movement and proliferation, in order to have a clearer picture of the modelling assumptions that are made and ensure that they correctly reflect the salient features of the underlying application problem.

For this reason, the derivation of continuum models formulated in terms of partial differential equations or partial integrodifferential equations from underlying individual-based models has attracted the attention of a considerable number of mathematicians and physicists. Examples in this active field of research include the derivation of continuum models of chemotaxis from velocity-jump process~\cite{hillen2009user,othmer1988models,othmer2000diffusion,painter2003modelling} or from self-attracting reinforced random walks~\cite{stevens2000derivation,stevens1997aggregation}; the derivation of diffusion and nonlinear diffusion equations from underlying random walks~\cite{champagnat2007invasion,inoue1991derivation,oelschlager1989derivation,othmer2002diffusion,penington2011building,penington2014interacting}, from systems of discrete equations of motion~\cite{oelschlager1990large,murray2009discrete,murray2012classifying}, from discrete lattice-based exclusion processes~\cite{binder2009exclusion,dyson2012macroscopic,fernando2010nonlinear,johnston2017co,johnston2012mean,landman2011myopic,lushnikov2008macroscopic,simpson2010cell} or from cellular automata~\cite{deroulers2009modeling,drasdo2005coarse,simpson2007simulating}; and, most recently, the derivation of nonlocal models of cell-cell adhesion from position-jump processes~\cite{buttenschoen2018space}. However, with the exception of the results presented in~\cite{byrne2009individual,engblom2018scalable,motsch2018short}, little prior work has investigated the relation between single-cell-based models and continuum models in the form of equation~\eqref{eq:pde1} and the system of equations~\eqref{eq:pde2}. In particular, a derivation of these continuum models of growing cell populations from underlying individual-based models of spatial cell movement and proliferation has remained elusive.

\paragraph{Contents of the paper.} In the present paper we aim to bridge such a gap in the existing literature. In Section~\ref{sec:discrete}, we develop a simple stochastic individual-based model for the dynamics of growing cell populations. Our model is based on a branching random walk that takes into account the effects of pressure-driven cell movement and pressure-dependent cell proliferation. In Section~\ref{sec:discreteder}, we show that equation~\eqref{eq:pde1} and the system of equations~\eqref{eq:pde2} can be formally derived from the branching random walk that underlies our discrete model. In Section~\ref{sec:quantcomp}, we carry out a systematic quantitative comparison between the individual-based model and its continuum counterparts, both in the case of one single cell population and in the case of multiple cell populations with different biophysical properties. In summary, we construct travelling-wave solutions both for equation~\eqref{eq:pde1} and for the system of equations~\eqref{eq:pde3} (Section~\ref{sec:twan}), we present numerical solutions that illustrate the results of the travelling-wave analysis, and we compare such numerical solutions with the results of computational simulations of the individual-based model (Section~\ref{sec:numres}). Section~\ref{sec:fin} concludes the paper and provides a brief overview of possible research perspectives.

\section{An individual-based model for growing cell populations}
\label{sec:discrete}
We consider a multicellular system composed of $M$ cell populations. We represent each cell within the system as an agent that occupies a position on a lattice. Cells can move and proliferate according to a set of simple rules that result in a discrete-time branching random walk. For ease of presentation, we let cells be arranged along the real line $\mathbb{R}$, but there would be no additional difficulty in considering the case of branching random walks in higher spatial dimensions. 

We discretise the time variable $t \in \mathbb{R}^+$ and the space variable $x \in \mathbb{R}$ as $t_{k} = k \tau$ with $k\in\mathbb{N}_{0}$ and $x_{i} = i  \chi$ with $i \in \mathbb{Z}$, respectively, where $0<\tau, \chi \ll 1$. We introduce the dependent variable $n^{k}_{hi}\in\mathbb{N}_0$ to model the number of cells of population $h=1, \ldots, M$ on the lattice site $i$ and at the time-step $k$, and we compute the density of cells of population $h$ and the total cell density, respectively, as
\begin{equation}
\label{e:n}
\rho_h(t_k,x_{i})= \rho^k_{hi} = n^{k}_{hi} \,  \chi^{-1} \; \text{ and } \; \rho(t_k,x_{i})= \rho^k_{i} = \sum_{h=1}^M \rho_h(t_k,x_{i}).
\end{equation}
Moreover, for each lattice site $i$ and time-step $k$ we assume the cell pressure $p(t_k,x_{i})= p^k_{i}$ to be given by a barotropic relation $p^k_{i} \equiv \Pi(\rho^k_{i})$ with $\Pi$ being a function of the total cell density that satisfies conditions~\eqref{as:p}. At each time-step, we allow every cell to undergo pressure-dependent proliferation and pressure-driven movement according to the following algorithmic rules, which are schematised in Figure~\ref{Fig1}.

\paragraph{Pressure-dependent cell proliferation.} 
We allow every cell to divide, die or remain quiescent with probabilities that depend on the local pressure, and we assume that a dividing cell is replaced by two identical daughter cells that are placed on the original lattice site of the parent cell. In order to model pressure-limited growth, given the net growth rate $G_h$ that satisfies assumptions~\eqref{as:G} for all values of $h$, we assume that at the $k^{th}$ time-step a focal cell of population $h$ on the lattice site $i$ can divide with probability 
\begin{equation}
\label{e:div}
\tau  \, G_h(p^{k}_{i})_+ \quad \mbox{where} \quad G_h(p^{k}_{i})_+ = \max\left(0,G_h(p^{k}_{i})\right)
\end{equation}
or die with probability 
\begin{equation}
\label{e:apo}
\tau  \, G_h(p^{k}_{i})_-  \quad \mbox{where} \quad G_h(p^{k}_{i})_- = - \min\left(0,G_h(p^{k}_{i})\right)
\end{equation}
or remain quiescent with probability 
\begin{equation}
\label{e:qui}
1 - \left(\tau  \, G_h(p^{k}_{i})_+ + \tau \, G_h(p^{k}_{i})_-\right) = 1 - \tau |G_h(p^{k}_{i})|.
\end{equation}
We assume the time-step $\tau$ to be sufficiently small so that the quantities~\eqref{e:div}-\eqref{e:qui} are all between 0 and 1. Under assumptions~\eqref{as:G}, the definitions \eqref{e:div}-\eqref{e:qui} are such that if $p^{k}_{i} > P$ then every cell on the $i^{th}$ lattice site can only die or remain quiescent at the $k^{th}$ time-step. Therefore, we have that
\begin{equation}
\label{eq:UBdisc} 
p^k_{i} \leq \overline{p} \; \mbox{ for all } \; (k,i) \in \mathbb{N}_0 \times \mathbb{Z}, \; \text{ with } \; \overline{p} = \max \left(\max_{i \in \mathbb{Z}} p^0_i, P \right).
\end{equation}

\paragraph{Pressure-driven cell movement.} We model pressure-driven cell movement (\emph{i.e.} the tendency of cells to move down pressure gradients) as a biased random walk whereby the movement probabilities depend on the difference between the pressure at the site occupied by a cell and the pressure at the neighbouring sites. In particular, for a focal cell of population $h$ on the lattice site $i$ at the time-step $k$, we define the probability of moving to the lattice site $i-1$ (\emph{i.e.} the probability of moving left) as
\begin{equation}
\label{e:left}
J^{L}_h(p^{k}_{i}-p^{k}_{i-1}) =  \nu_h \frac{(p^{k}_{i}-p^{k}_{i-1})_{+}}{2 \, \overline{p}}, 
\end{equation}
where $(p^{k}_{i}-p^{k}_{i-1})_{+} = \max\left(0,p^{k}_{i}-p^{k}_{i-1}\right)$, the probability of moving to the lattice site $i+1$ (\emph{i.e.} the probability of moving right) as
\begin{equation}
\label{e:right}
J^{R}_h(p^{k}_{i}-p^{k}_{i+1}) = \nu_h \frac{(p^{k}_{i}-p^{k}_{i+1})_{+}}{2 \, \overline{p}},  
\end{equation}
where $(p^{k}_{i}-p^{k}_{i+1})_{+} = \max\left(0,p^{k}_{i}-p^{k}_{i+1}\right)$, and the probability of remaining stationary on the lattice site $i$ as 
\begin{equation}
\label{e:stay}
1 - J^{L}_h(p^{k}_{i}-p^{k}_{i-1}) - J^{R}_h(p^{k}_{i}-p^{k}_{i+1}).
\end{equation}
In the above equations, the coefficient $0 < \nu_h \leq 1$ is directly proportional to the mobility of cells in population $h$ and the parameter $\overline{p}$ is defined in~\eqref{eq:UBdisc}. Notice that the definitions~\eqref{e:left}-\eqref{e:stay} are such that the cells will move down pressure gradients. Moreover, the a priori estimate~\eqref{eq:UBdisc} ensures that the quantities defined according to~\eqref{e:left}-\eqref{e:stay} are all between 0 and 1. 

\begin{figure}
\centering
\includegraphics[width=1\textwidth]{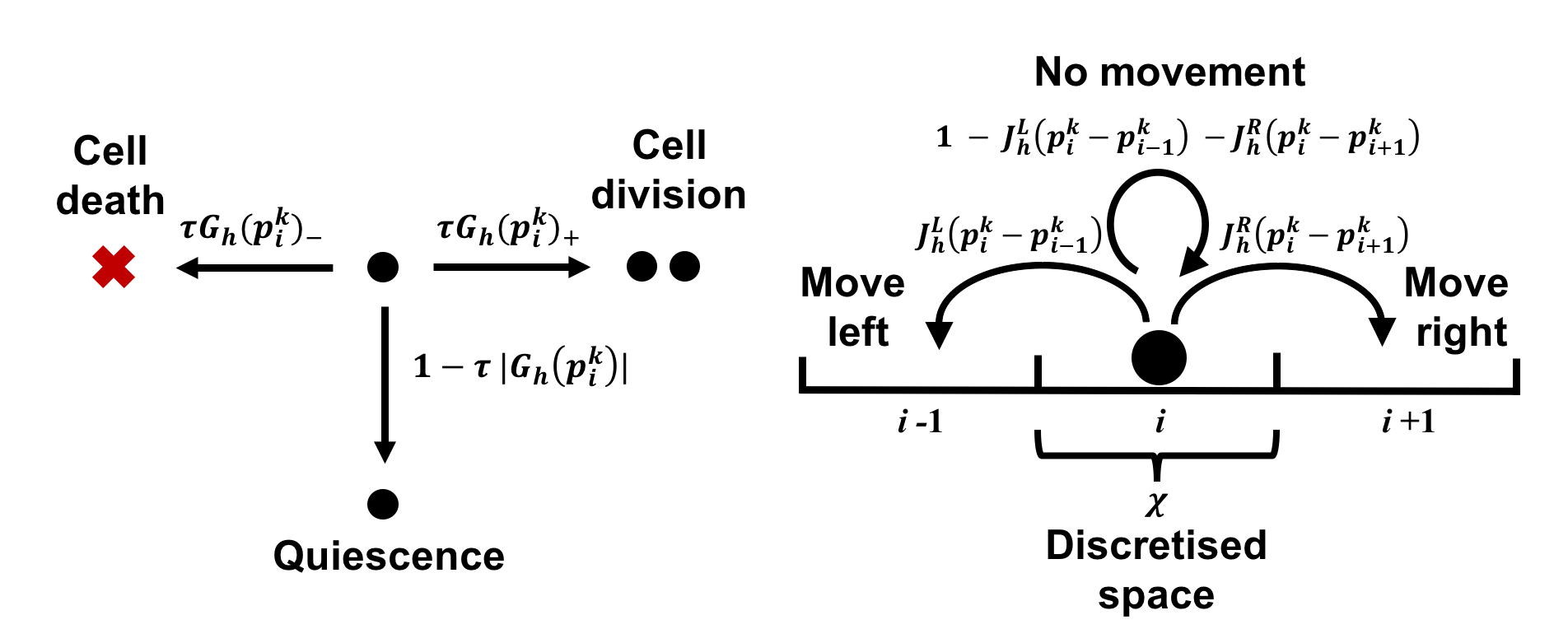}
\caption{Schematic representation of the algorithmic rules governing cell dynamics in our stochastic individual-based model. Pressure-dependent cell proliferation is modelled by letting the probabilities of a cell dividing, dying and remaining quiescent depend on the pressure at the site occupied by the cell (left panel). Pressure-driven cell movement is modelled by letting the movement probabilities depend on the difference between the pressure at the site occupied by a cell and the pressure at the neighbouring sites (right panel)}
\label{Fig1}
\end{figure}

\section{Formal derivation of continuum models}
\label{sec:discreteder}
In this section, we show how continuum models of growing cell populations of the form of equation \eqref{eq:pde1} and of the system of equations \eqref{eq:pde2} and \eqref{eq:pde3} can be derived as formal limits of the branching random walk that underlies our stochastic individual-based model. 

For a multicellular system the dynamic of which is governed by the algorithmic rules presented in Section \ref{sec:discrete}, the principle of mass balance gives 
\begin{eqnarray*}
\rho^{k+1}_{hi}  &=& \nu_h  \frac{(p^{k}_{i-1}-p^{k}_{i})_{+}}{2 \, \overline{p}} \left[2  \tau  G_h(p^{k}_{i-1})_+ + \left(1 - \tau  |G_h(p^{k}_{i-1})| \right) \right]  \rho^{k}_{h \, i-1}  \\
&& + \nu_h  \frac{(p^{k}_{i+1}-p^{k}_{i})_{+}}{2 \, \overline{p}} \left[2  \tau  G_h(p^{k}_{i+1})_+ + \left(1 - \tau  |G_h(p^{k}_{i+1})| \right) \right]  \rho^{k}_{h \, i+1} \\
&& + \left[1 - \nu_h  \frac{(p^{k}_{i}-p^{k}_{i-1})_{+}}{2 \, \overline{p}} - \nu_h  \frac{(p^{k}_{i}-p^{k}_{i+1})_{+}}{2 \, \overline{p}}\right] \\
&& \phantom{[1 - \nu_h  \frac{(p^{k}_{i}-p^{k}_{i-1})_{+}}{P} -aaaa} \times \Big[2  \tau  G_h(p^{k}_{i})_+ + \left(1 - \tau  |G_h(p^{k}_{i})| \right) \Big]  \rho^{k}_{hi}
\end{eqnarray*}
and after a little algebra we find
\begin{eqnarray*}
\rho^{k+1}_{hi}  &=& \nu_h  \frac{(p^{k}_{i-1}-p^{k}_{i})_{+}}{2 \, \overline{p}} \left(\tau  G_h(p^{k}_{i-1}) + 1 \right)  \rho^{k}_{h \, i-1} \\
&&+ \nu_h  \frac{(p^{k}_{i+1}-p^{k}_{i})_{+}}{2 \, \overline{p}} \left(\tau  G_h(p^{k}_{i+1}) + 1 \right)  \rho^{k}_{h \, i+1} \\
&& + \left[1 - \nu_h  \frac{(p^{k}_{i}-p^{k}_{i-1})_{+}}{2 \, \overline{p}} - \nu_h  \frac{(p^{k}_{i}-p^{k}_{i+1})_{+}}{2 \, \overline{p}}\right]  \left(\tau  G_h(p^{k}_{i}) + 1 \right)  \rho^{k}_{hi}.
\end{eqnarray*}
The above equation simplifies to
\begin{eqnarray}
\label{eq:intera1}
\rho^{k+1}_{hi} - \rho^{k}_{hi} &=& \tau  G_h(p^{k}_{i})  \rho^{k}_{hi} \nonumber \\
&& +\frac{\nu_h}{2 \, \overline{p}} \left[\rho^{k}_{h \, i-1} (p^{k}_{i-1}-p^{k}_{i})_{+} + \rho^{k}_{h \, i+1} (p^{k}_{i+1}-p^{k}_{i})_{+} \right] \nonumber \\
&& - \frac{\nu_h}{2 \, \overline{p}}\left[\rho^{k}_{hi} (p^{k}_{i}-p^{k}_{i-1})_{+}  + \rho^{k}_{hi} (p^{k}_{i}-p^{k}_{i+1})_{+}\right] \nonumber\\
&& + \frac{ \nu_h \tau}{2 \, \overline{p}} \left[\rho^{k}_{h \, i-1} G_h(p^{k}_{i-1}) (p^{k}_{i-1}-p^{k}_{i})_{+}  + \rho^{k}_{h \, i+1} G_h(p^{k}_{i+1}) (p^{k}_{i+1}-p^{k}_{i})_{+}\right]\nonumber\\
&&  -  \frac{ \nu_h \tau}{2 \, \overline{p}}  \left[\rho^{k}_{hi} G_h(p^{k}_{i}) (p^{k}_{i}-p^{k}_{i-1})_{+} + \rho^{k}_{hi} G_h(p^{k}_{i}) (p^{k}_{i}-p^{k}_{i+1})_{+} \right].
\end{eqnarray}
Using the fact that the following relations hold for $\tau$ and $\chi$ sufficiently small
$$
t _k \approx t, \quad t _{k+1} \approx t + \tau, \quad x_{i} \approx x, \quad x_{i \pm 1} \approx x \pm \chi,
$$
$$
\rho^k_{h i} \approx \rho_h(t,x), \quad \rho^{k+1}_{h i} \approx \rho_h(t+\tau,x), \quad \rho^k_{h \, i \pm 1} \approx \rho_h(t,x \pm \chi),
$$
$$
p^k_i \approx p(t,x), \quad p^k_{i \pm 1} \approx p(t,x \pm \chi),
$$
we rewrite equation \eqref{eq:intera1} in the approximate form
\begin{eqnarray}
\label{eq:interb1}
\rho_h(t+\tau,x) - \rho_h(t,x) &\approx& \tau  G_h(p(t,x)) \rho_h(t,x) \nonumber\\
&& +\frac{\nu_h}{2 \, \overline{p}} \left[\rho_h(t,x-\chi) (p(t,x-\chi)-p(t,x))_{+} \right] \nonumber\\
&& +\frac{\nu_h}{2 \, \overline{p}} \left[ \rho_h(t,x+\chi) (p(t,x+\chi)-p(t,x))_{+}\right] \nonumber\\
&& - \frac{\nu_h}{2 \, \overline{p}}\left[\rho_h(t,x) (p(t,x)-p(t,x-\chi))_{+}\right] \nonumber\\  
&& - \frac{\nu_h}{2 \, \overline{p}}\left[\rho_h(t,x) (p(t,x)-p(t,x+\chi))_{+} \right]\nonumber\\
&& + \frac{ \nu_h \tau}{2 \, \overline{p}} \left[\rho_h(t,x-\chi) G_h(p(t,x-\chi)) (p(t,x-\chi)-p(t,x))_{+} \right]\nonumber\\
&&  + \frac{ \nu_h \tau}{2 \, \overline{p}} \left[\rho_h(t,x+\chi) G_h(p(t,x+\chi)) (p(t,x+\chi)-p(t,x))_{+} \right]\nonumber\\
&&  -  \frac{ \nu_h \tau}{2 \, \overline{p}}  \left[\rho_h(t,x)  G_h(p(t,x)) (p(t,x)-p(t,x-\chi))_{+} \right]\nonumber\\
 &&  -  \frac{ \nu_h \tau}{2 \, \overline{p}}  \left[\rho_h(t,x) G_h(p(t,x)) (p(t,x)-p(t,x+\chi))_{+}\right].
\end{eqnarray}
Assuming
\beq
\label{assC2}
\rho_h \in C^2(\mathbb{R}^+ \times \, \mathbb{R}), \quad h=1, \ldots, M,
\eeq  
we approximate the terms $\rho_h(t+\tau,x)$, $\rho_h(t,x-\chi)$ and $\rho_h(t,x+\chi)$ in equation~\eqref{eq:interb1} by their second order Taylor expansions about the point $(t,x)$. Moreover, since $p \equiv \Pi(\rho)$ and $\Pi$ is a smooth function of $\rho$, under assumption~\eqref{assC2} the pressure $p(t,x)$ is twice continuously differentiable with respect to the variable $x$. Hence we also approximate the terms $p(t,x-\chi)$ and $p(t,x+\chi)$ in equation \eqref{eq:interb1} by their second order Taylor expansions about the point $(t,x)$. In so doing, after a little algebra we find
\begin{eqnarray*}
\tau \, \partial_{t}\rho_h(t,x)+\frac{\tau^2}{2}\partial^2_{tt}\rho_h(t,x) &\approx& \tau  G_h(p(t,x)) \rho_h(t,x) + \frac{\nu_h \chi^{2}}{2 \, \overline{p}} \rho_h(t,x) \partial^2_{xx} p(t,x)\nonumber\\
&& + \frac{\nu_h \chi^2}{2 \, \overline{p}} \left[ \left(\partial_{x}p(t,x)\right)_{+}-\left(-\partial_{x}p(t,x)\right)_{+}\right] \partial_{x} \rho_h(t,x) \nonumber\\
&& + \frac{ \nu_h \tau}{2 \, \overline{p}} \rho_h(t,x) G_h(p(t,x-\chi)) \left(-\chi \partial_{x}p(t,x)\right)_{+} \nonumber\\
&& + \frac{ \nu_h \tau}{2 \, \overline{p}} \rho_h(t,x) G_h(p(t,x+\chi)) \left(\chi \partial_{x}p(t,x)\right)_{+} \nonumber\\
&&- \frac{ \nu_h \tau}{2 \, \overline{p}} \rho_h(t,x) G_h(p(t,x)) \left(\chi\partial_{x}p(t,x)\right)_{+} \nonumber\\
&&- \frac{ \nu_h \tau}{2 \, \overline{p}} \rho_h(t,x) G_h(p(t,x))\left(-\chi\partial_{x}p(t,x)\right)_{+},
\end{eqnarray*}
which implies
 \begin{eqnarray}
\label{equation7a}
\tau \, \partial_{t}\rho_h(t,x)+\frac{\tau^2}{2}\partial^2_{tt}\rho_h(t,x) &\approx& \tau  G_h(p(t,x)) \rho_h(t,x) 
\nonumber \\
&& + \frac{\nu_h \chi^{2}}{2 \, \overline{p}} \left(\rho_h(t,x) \partial^2_{xx}p(t,x) + \partial_{x}\rho_h(t,x) \partial_{x}p(t,x) \right) 
\nonumber \\
&& + \frac{\nu_h\tau\chi}{2 \, \overline{p}} \, F(t,x), 
\end{eqnarray}
with
 \begin{eqnarray*}
F(t,x) &=& \left[G_h(p(t,x-\chi)) \left(-\partial_{x}p(t,x)\right)_{+} + G_h(p(t,x+\chi))\left(\partial_{x}p(t,x)\right)_{+} \right]\rho_h(t,x) 
\nonumber \\
&& - \left[\left(\partial_{x}p(t,x)\right)_{+} + \left(-\partial_{x}p(t,x)\right)_{+} \right] G_h(p(t,x)) \rho_h(t,x).
 \end{eqnarray*}
Dividing both sides of the resulting equation by $\tau$ we obtain 
 \begin{eqnarray}
\label{equation7}
\partial_{t}\rho_h(t,x)+\frac{\tau}{2}\partial^2_{tt}\rho_h(t,x)  &\approx& G_h(p(t,x)) \rho_h(t,x) 
\nonumber\\
&& + \frac{\nu_h \chi^{2}}{2 \, \overline{p} \tau} \left(\rho_h(t,x) \partial^2_{xx}p(t,x) + \partial_{x}\rho_h(t,x) \partial_{x}p(t,x) \right)\nonumber\\
&& + \frac{ \nu_h\chi}{2 \, \overline{p}} \, F(t,x).
\end{eqnarray}
Letting both $\tau \to 0$ and $\chi \to 0$ in such a way that
\beq
\label{asymat}
\frac{\nu_h \chi^{2}}{2 \, \overline{p} \, \tau} \rightarrow \mu_h \quad \text{as } \; \tau \to 0 \; \text{ and } \; \chi \to 0, \quad \text{for} \quad h=1,\ldots,M,
\eeq
from~\eqref{equation7} we formally obtain the following system of coupled conservation equations
$$
\partial_{t}\rho_h=  G_h(p)  \rho_h + \mu  \left(\rho_h \, \partial^2_{xx} p_h + \partial_{x}\rho_h \,  \partial_{x}p_h \right), \quad h=1,\ldots,M,
$$
which can be rewritten as the system of equations~\eqref{eq:pde2}, that is, 
\beq
\label{equation8}
\partial_{t}\rho_h - \mu_h \, \partial_x \left(\rho_h \, \partial_{x}p \right) = G_h(p)  \rho_h, \quad h=1,\ldots,M.
\eeq
In a similar way, in the case of one single cell population, letting $M=1$ and dropping the index $h$ we formally obtain equation~\eqref{eq:pde1}. Moreover, we can formally obtain the system of equations~\eqref{eq:pde3} by choosing $M=2$, labelling the two populations by $h=1$ and $h=2$, and setting $G_1 \equiv G$ and $G_2 \equiv 0$.

\begin{remark}
We note that condition~\eqref{asymat} is a natural counterpart of the usual parabolic scaling of Brownian motion. Hence, our formal derivation does not impose any additional assumptions than those commonly employed for the asymptotic investigation of random walks.
\end{remark}

\section{Comparison between individual-based and continuum models}
\label{sec:quantcomp}
In this section, we carry out a systematic quantitative comparison between the outcomes of our individual-based model, both in the case of one cell population and in the case of two cell populations, and the solutions of the corresponding continuum models. We first establish the existence of travelling-wave solutions for the continuum models~\eqref{eq:pde1} and~\eqref{eq:pde3} (Section~\ref{sec:twan}). We then construct numerical solutions of the model equations which illustrate the results of the travelling-wave analysis and we compare such numerical solutions with the results of computational simulations of our individual-based model (Section~\ref{sec:numres}). 

\subsection{Travelling-wave analysis of the continuum models}
\label{sec:twan}
We first consider the continuum model~\eqref{eq:pde1} and we look for one-dimensional travelling-wave solutions of the form 
$$
\rho(t,x)=\rho(z) \quad \text{with} \quad z=x-ct \quad \text{and} \quad c>0
$$
that satisfy the following asymptotic conditions
\beq
\label{ass:TW1}
\rho(z) \; \xrightarrow[z  \rightarrow -\infty]{} \Pi^{-1}(P) \quad \text{and} \quad \rho(z) \; \xrightarrow[z  \rightarrow \infty]{} 0.
\eeq 
Therefore, we study the existence of pairs $(\rho,c)$ that satisfy the problem defined by the differential equation 
\beq
\label{eq:TW1}
- c \, \rho' - \mu \left(\rho \, p' \right)' = G(p) \, \rho
\eeq 
subject to conditions~\eqref{ass:TW1}. Our main results are summarised by the following theorem.
\begin{theorem}
\label{th:TW1}
Under assumptions~\eqref{as:G} and \eqref{as:p}, there exists $c>0$ such that the travelling-wave problem defined by the differential equation \eqref{eq:TW1} subject to conditions~\eqref{ass:TW1} admits a nonnegative and nonincreasing solution $\rho$.
\end{theorem}

\begin{proof}
We divide the proof of Theorem~\ref{th:TW1} into two steps. First we prove that for $c>0$ fixed the differential equation~\eqref{eq:TW1} subject to the asymptotic conditions~\eqref{ass:TW1} admits a nonnegative and nonincreasing solution $\rho$ (Step~1). Then we show that there exists a unique value of the wave speed $c$ that satisfies the travelling-wave problem (Step~2). 
\\\\
{\it Step 1.} Multiplying both sides of the differential equation~\eqref{eq:TW1} by $\displaystyle{\frac{\rmd p}{\rmd \rho}}$ we obtain the following boundary-value problem for $p$ 
\beq
\label{e.P-inf0FB1}
- p' \, \left(c + \mu \, p' \right) - \mu \, p'' \, \rho \, \frac{\rmd p}{\rmd \rho} = G\left(p\right) \rho \frac{\rmd p}{\rmd \rho},
\eeq
\beq
\label{e.P-inf0BCFB1}
p(z) \xrightarrow[ z \rightarrow - \infty]{} P \quad \mbox{and} \quad p(z) \xrightarrow[ z \rightarrow \infty]{} 0.
\eeq
Let $z^*$ be a critical point of $p$ in $\mathbb{R}$. Using the differential equation \eqref{e.P-inf0FB1} we see that 
$$
p''(z^*) = - G\left(p(z^*)\right)
$$
and, under assumptions~\eqref{as:G} and conditions~\eqref{e.P-inf0BCFB1}, using the strong maximum principle we conclude that $p < P$ in $\mathbb{R}$ and that $p$ cannot have a local minimum in $\mathbb{R}$, \emph{i.e.}
\beq
\label{e.Pon-inf0FB}
p'(z) < 0 \; \mbox{ for all } \; z \in \mathbb{R}.
\eeq 
Hence the solution $p$ of the differential equation~\eqref{e.P-inf0FB1} subject to conditions~\eqref{e.P-inf0BCFB1} is a nonnegative and nonincreasing function that satisfies 
\beq
\label{e.Pon-inf0FB1}
0 <  p(z) < P \; \mbox{ for all } \; z \in \mathbb{R}.
\eeq 
Since $p \equiv \Pi(\rho)$ and $\Pi$ is a smooth monotonically increasing function of $\rho$, we can then conclude that the cell density $\rho$ is a nonnegative and nonincreasing function as well. 
\\\\
{\it Step 2.} Using a method similar to that used in {\it Step 5} of the proof of Theorem~\ref{th:TW2} one can prove that $p$ is a monotonically decreasing function of the parameter $c$. This ensures that given the pressure $p$ or, equivalently, the cell density $\rho$ the wave speed $c$ can be uniquely identified through a monotonicity argument.
\qed
\end{proof}

We then turn to the travelling-wave analysis of the system of equations~\eqref{eq:pde3}. On the basis of the results presented in~\cite{lorenzi2017interfaces}, we consider one-dimensional travelling-wave solutions of the form 
$$
\rho_1(t,x)=\rho_1(z) \quad \text{and} \quad \rho_2(t,x)=\rho_2(z), \quad \text{with} \quad z=x-ct \quad \text{and} \quad c>0,
$$
that satisfy the following conditions
\beq
\label{ass:TW2a}
\rho_1(z)\left\{
\begin{array}{ll}
> 0, \quad \text{for } z < 0,
\\\\
= 0 , \quad \text{for } z \geq 0,
\end{array}
\right.
\qquad
\rho_2(z)\left\{
\begin{array}{ll}
= 0 , \quad \text{for } z < 0,
\\\\
> 0, \quad \text{for } z \in [0, \ell),
\\\\
= 0 , \quad \text{for } z \geq \ell,
\end{array}
\right.
\eeq
for some $\ell > 0$, along with the asymptotic condition 
\beq
\label{ass:TW2b}
\rho_1(z) \; \xrightarrow[z  \rightarrow -\infty]{} \Pi^{-1}(P).
\eeq
Hence, we study the existence of triples $(\rho_1,\rho_2,c)$ along with $\ell > 0$ that satisfy the problem defined by the following system of differential equations
\beq
\label{eq:TW2}
\left\{
\begin{array}{ll}
-c \, \rho'_1 - \mu_1 \left(\rho_1 \, p' \right)' = G(p) \rho_1, 
\\\\
-c \, \rho'_2 - \mu_2 \left(\rho_2 \, p' \right)' = 0,
\end{array}
\right.
\eeq
subject to conditions~\eqref{ass:TW2a} and~\eqref{ass:TW2b}. Notice that the principle of mass conservation gives
\beq
\label{eq:mascon}
\int_0^{\ell} \rho_2(z) \; {\rmd}z= N_2
\eeq 
for some $N_2>0$ that represents the number of cells in population $2$. Our main results are summarised by the following theorem.

\begin{theorem}
\label{th:TW2}
Under assumptions~\eqref{as:G} and \eqref{as:p}, for any $N_2>0$ given, there exists $c>0$ and $\ell>0$ such that the system of differential equations \eqref{eq:TW2} subject to conditions~\eqref{ass:TW2a} and \eqref{ass:TW2b} admits a component-wise nonnegative solution $(\rho_1,\rho_2)$ with $\rho_1$ nonincreasing, and $\rho_2$ nonincreasing and satisfying condition~\eqref{eq:mascon}. Moreover, the pressure $p$ has a kink in $z=0$ with 
\beq
\label{eq:ppjump}
\sgn(p'(0^+) - p'(0^-)) =\sgn(\mu_2 - \mu_1).
\eeq 
\end{theorem}

\begin{proof}
Building upon the method of proof presented by Lorenzi {\it et al}~\cite{lorenzi2017interfaces} for the case of the barotropic relation~\eqref{def:p}, we prove Theorem~\ref{th:TW2} in five steps. We fix the parameter $c>0$ and first prove that, for $N_2>0$ given, the problem under study admits a component-wise nonnegative solution $(\rho_1,\rho_2)$ with $\rho_2$ nonincreasing and with the value of $\ell$ being determined by condition~\eqref{eq:mascon} (Step 1), and with $\rho_1$ nonincreasing (Step 2). Then we prove that the total cell density $\rho$ is continuous on $(-\infty, \ell)$ (Step 3) and the jump condition~\eqref{eq:ppjump} holds (Step 4). Finally, we show that there exists a unique value of the wave speed $c$ that satisfies the travelling-wave problem (Step 5). 
\\\\
{\it Step 1.} Integrating the differential equation \eqref{eq:TW2}$_2$ between a generic point $z \in [0,\ell]$ and $\infty$, and using the fact that both $p'(z) \to 0$ and $\rho_2(z) \to 0$ as $z \to \infty$, we find
\beq
\label{e.Pon0rFB22}
p'(z) = - \frac{c}{\mu_2} < 0 \; \mbox{ for all } \; z \in [0,\ell].
\eeq
Integrating equation \eqref{e.Pon0rFB22} between a generic point $z \in [0,\ell)$ and $\ell$, and using the fact that $p(\ell)=0$, gives
\beq
\label{e:PeqFB}
p(z) = \frac{c}{\mu_2} \, (\ell - z) \; \mbox{ for } \; z \in [0,\ell],
\eeq
which implies that
\beq
\label{e.Pon0rFB12fm}
p(0) = \frac{c \, \ell}{\mu_2}.
\eeq
Since $\rho_1 \equiv 0$ on $[0,\ell]$, under assumptions~\eqref{as:p}, we have that $p$ is a monotonically decreasing function of $\rho_1$ in $[0,\ell]$. Hence, the results~\eqref{e.Pon0rFB22} and \eqref{e:PeqFB} allow us to conclude that $\rho_2$ is decreasing in $(0,\ell)$. Moreover, for $N_2>0$ given, since the value of $\rho_2(z)$ is uniquely determined for all $z \in [0,\ell]$, the value of $\ell$ is uniquely fixed by the mass conservation condition~\eqref{eq:mascon}. 
\\\\
\noindent {\it Step 2.} Since $\rho_2 \equiv 0$ on $(-\infty,0)$ and, therefore, $\rho_1 \equiv \rho$ on $(-\infty,0)$, multiplying both sides of the differential equation~\eqref{eq:TW2}$_1$ by $\displaystyle{\frac{\rmd p}{\rmd \rho}}$ we obtain the following boundary-value problem for $p$ 
\beq
\label{e.P-inf0FB12pop}
- p' \, \left(c + \mu_1 \, p' \right) - \mu_1 \, p'' \, \rho \, \frac{\rmd p}{\rmd \rho} = G\left(p\right) \rho \frac{\rmd p}{\rmd \rho},
\eeq
\beq
\label{e.P-inf0BCFB12pop}
p(z) \xrightarrow[ z \rightarrow - \infty]{} P \quad \mbox{and} \quad p(0) = \frac{c \, \ell}{\mu_2}.
\eeq
Hence, using a method similar to that used in {\it Step 1} of the proof of Theorem~\ref{th:TW1} one can prove that $\rho_1$ is decreasing in $(-\infty,0)$. 
\\\\
\noindent {\it Step 3.} The results proved in {\it Step 1} and {\it Step 2} ensure that the total cell density $\rho$ is nonincreasing and continuous in $(-\infty,0)$ and in $(0,\ell)$. We now prove that $\rho$ is continuous in $z=0$. Adding together the differential equations~\eqref{eq:TW2}$_1$~and~\eqref{eq:TW2}$_2$ gives
\beq
\label{eq:sumrhostw2}
-c \left(\rho_1 + \rho_2 \right)' - \big(\left(\mu_1 \rho_1 + \mu_2 \rho_2\right) \, p' \big)' = G(p) \rho_1.
\eeq
Multiplying both sides of the above differential equation by $p$ and using the fact that
$$
p \, \big(\left(\mu_1 \rho_1 + \mu_2 \rho_2\right) \, p' \big)' = \big(p \, \left(\mu_1 \rho_1 + \mu_2 \rho_2\right) \, p' \big)' - \left(\mu_1 \rho_1 + \mu_2 \rho_2\right) \, (p')^2
$$
we find that
$$
\left(\mu_1 \rho_1 + \mu_2 \rho_2\right) \, (p')^2 = G(p) \rho_1 + c \, p \left(\rho_1 + \rho_2 \right)' + \big(p \, \left(\mu_1 \rho_1 + \mu_2 \rho_2\right) \, p' \big)'. 
$$
Integrating both sides of the latter differential equation between a generic point $z^*<0$ and $\ell$, and estimating the right-hand side from above by using the fact that $-\infty<\left(\rho_1 + \rho_2 \right)'(z) \leq 0$ for all $z \in [z^*,\ell)$, $p(\ell)=0$ and $\rho_2(z^*)=0$, yields
$$
\int_{z^*}^{\ell}  \left(\mu_1 \rho_1 + \mu_2 \rho_2\right) \, (p')^2 \, {\rmd}z \leq \int_{z^*}^{\ell} G(p) \rho_1 \, {\rmd}z - \mu_1 \, \rho_1(z^*) \, p(z^*) \, p'(z^*) < \infty.
$$
The above integral inequality ensures that $p' \in L^2_{loc}(\mathbb{R})$. This result along with the fact that $p \in L^{\infty}(\mathbb{R})$ allow us to conclude that $p$ is continuous in $z=0$. Since $p \equiv \Pi(\rho)$ and $\Pi$ is a smooth monotonically increasing function of $\rho$, we have that the total cell density $\rho$ is continuous in $z=0$ as well.
\\\\
\noindent {\it Step 4.} Integrating the differential equation \eqref{eq:sumrhostw2} between a generic point $z < \ell$ and $\ell$ and using the fact that $\rho_1(\ell)=\rho_2(\ell)=0$ yields
$$
c \left(\rho_1(z) + \rho_2(z) \right) + \left(\mu_1 \rho_1(z) + \mu_2 \rho_2(z)\right) \, p'(z) = \int_{z}^{\ell} G(p) \rho_1 \, {\rmd}z'.
$$
Letting $z \to 0^-$ and using the fact that $\rho_1(0^-) = \rho(0^-)$ and $\rho_1 \equiv 0$ on $[0,\ell]$ we find that
\beq
\label{neqrev1}
c \, \rho(0^-) + \mu_1 \, \rho(0^-) \, p'(0^-) = 0.
\eeq
Similarly, letting $z \to 0^+$ and using the fact that $\rho_2(0^+) = \rho(0^+)$ and $\rho_1 \equiv 0$ on $[0,\ell]$ gives
\beq
\label{neqrev2}
c \, \rho(0^+) + \mu_2 \, \rho(0^+) \, p'(0^+) = 0.
\eeq
Since $\rho(z)$ is continuous in $z=0$, combining equations~\eqref{neqrev1} and~\eqref{neqrev2} we obtain
$$
\mu_1 \, p'(0^-) = \mu_2 \, p'(0^+) \quad \Longrightarrow \quad p'(0^-) = \frac{\mu_2}{\mu_1} \, p'(0^+).
$$
This result along with the expression~\eqref{e.Pon0rFB22} for $p'(0^+)$ gives
$$
p'(0^-) = -\frac{c}{\mu_1} \quad \text{and} \quad p'(0^+) - p'(0^-) = \frac{c}{\mu_1 \, \mu_2} \left(\mu_2 - \mu_1\right).
$$
From the latter equation we deduce the jump condition~\eqref{eq:ppjump}. Moreover, substituting the expression of $p'(0^-)$ into the differential equation~\eqref{e.P-inf0FB12pop} gives $p''(0^-)<0$.
\\\\
\noindent {\it Step 5.}  We prove that $p$ is a monotonically decreasing function of the parameter $c$ on $(-\infty,0)$. This ensures that given the cell density $\rho_1$ the wave speed $c$ can be uniquely identified through a monotonicity argument. Recalling that $\rho_2 \equiv 0$ on $(-\infty,0)$ and, therefore, $\rho_1 \equiv \rho$ on $(-\infty,0)$, differentiating equation~\eqref{eq:TW2}$_1$ with respect to $z$ we find
\beq 
\label{e.mono2FB1}
- c \, \left(\rho'\right)' - \mu_1 \left[\left(p'\right)'' \rho + p'' \, \rho' + (p')' \rho' + p'  (\rho')' \right] = \frac{\rmd G}{\rmd p} \,p' \, \rho + G\left(p\right) \, \rho'
\eeq
with
$$
p'(z) \xrightarrow[ z \rightarrow - \infty]{} 0 \quad \text{and} \quad p'(0) = -\frac{c}{\mu_1}.
$$
On the other hand, differentiating equation~\eqref{eq:TW2}$_1$ with respect to $c$ gives
\begin{equation}  \label{e.mono1FB1}
\begin{aligned} 
- c \, \left(\frac{\partial \rho}{\partial c}\right)' - \mu_1 \left[\left(\frac{\partial p}{\partial c}\right)'' \rho + p'' \frac{\partial \rho}{\partial c} + \left(\frac{\partial p}{\partial c}\right)' \rho' + p'  \left(\frac{\partial \rho}{\partial c}\right)' \right]  \\
 =  \frac{\rmd G}{\rmd p}  \, \frac{\partial p}{\partial c} \, \rho + G\left(p\right) \frac{\partial \rho}{\partial c} + \rho'
\end{aligned} 
\end{equation} 
with
$$
\frac{\partial p}{\partial c}(z) \xrightarrow[ z \rightarrow - \infty]{} 0 \quad \text{and} \quad \left(\frac{\partial p}{\partial c}\right)'(0) = -\frac{1}{\mu_1}.
$$
Using the fact that $\displaystyle{p' =  \frac{\rmd p}{\rmd \rho} \, \rho'}$, we rewrite the differential equations \eqref{e.mono2FB1} and \eqref{e.mono1FB1}, respectively, as
\beq \label{e.mono2bisFB41}
 - c \, \left(\rho'\right)' - \mu \,\left[\left(p'\right)'' \rho + p'' \rho' + (p')' \rho' + \frac{\rmd p}{\rmd \rho}  \, \rho' \, (\rho')' \right] = \frac{\rmd G}{\rmd p}  \,p' \, \rho + G\left(p\right) \, \rho'
\eeq
and
\begin{eqnarray}  \label{e.mono1bisFB511}
\begin{aligned} 
- c \, \left(\frac{\partial \rho}{\partial c}\right)' - \mu \,\left[\left(\frac{\partial p}{\partial c}\right)'' \rho + p'' \frac{\partial \rho}{\partial c} + \left(\frac{\partial p}{\partial c}\right)' \rho' + \frac{\partial p}{\partial \rho} \, \rho' \, \left(\frac{\partial \rho}{\partial c}\right)' \right]  \\
=\frac{\rmd G}{\rmd p}  \, \frac{\partial p}{\partial c} \, \rho + G\left(p\right) \frac{\partial \rho}{\partial c} + \rho'.
\end{aligned} 
\end{eqnarray} 
Since $\displaystyle{\frac{\rmd p}{\rmd \rho}  > 0}$, we have 
$$
\rho' =  p' \, \left(\frac{\rmd p}{\rmd \rho} \right)^{-1} \; \mbox{and } \; \frac{\partial \rho}{\partial c} =  \frac{\partial p}{\partial c} \, \left(\frac{\rmd p}{\rmd \rho} \right)^{-1}.
$$
Hence, introducing the notation $f = p'$, $\displaystyle{g = \frac{\partial p}{\partial c}}$ and
$$
\begin{aligned} 
& k_0 = \left(\frac{\rmd p}{\rmd \rho} \right)^{-1}, \quad && k_1 = \rho,  \qquad k_2 = p'' \, \left(\frac{\rmd p}{\rmd \rho} \right)^{-1}, \quad && k_3 = \rho' , \\
& k_4 =\rho' \, \frac{\rmd p}{\rmd \rho},  \quad && k_5 = \frac{\rmd G}{\rmd p}  \, \rho + G\left(p\right) \, \left(\frac{\rmd p}{\rmd \rho} \right)^{-1},
\end{aligned} 
$$
we rewrite the differential equations \eqref{e.mono2bisFB41} and \eqref{e.mono1bisFB511}, respectively, as
\beq \label{e.mono2bisFBb61}
\displaystyle{- c \, \left( k_0 \, f \right)' - \mu \, \left[k_1 \, f'' + k_2 \, f + k_3 \, f' + k_4 \, \left(k_0 \, f\right)' \right] = k_5 \,f}
\eeq
with
\beq \label{e.BCnew1}
f(z) \xrightarrow[ z \rightarrow - \infty]{} 0, \quad f(0) = -\frac{c}{\mu_1}
\eeq
and 
\beq \label{e.mono1bisFBb71}
\displaystyle{- c \, \left( k_0 \, g \right)' - \mu \, \left[k_1 \, g'' + k_2 \, g + k_3 \, g' + k_4 \, \left(k_0 \, g\right)' \right] = k_5 \, g + \rho'}
\eeq
with
\beq \label{e.BCnew1b}
g(z) \xrightarrow[ z \rightarrow - \infty]{} 0, \quad g'(0) = -\frac{1}{\mu_1}.
\eeq
Since $f(z)<0$ for all $z \in (-\infty,0)$ and $f'(0)<0$, noting that both $f(z) \to 0$ and $g(z) \to 0$ as $z \to -\infty$ and the right-hand side of the differential equation~\eqref{e.mono1bisFBb71} contains the additional term $\rho'<0$ compared to the right-hand side of the differential equation~\eqref{e.mono2bisFBb61}, we deduce that $\displaystyle{g = \frac{\partial p}{\partial c} < 0}$, which concludes the proof of Theorem~\ref{th:TW2}.
\qed
\end{proof}

\begin{remark}
\label{reminsta}
Based on the jump condition~\eqref{eq:ppjump}, we expect the travelling-wave solution of Theorem~\ref{th:TW2} to be unstable if $\mu_1 >\mu_2$. In fact, a small perturbation of $\rho_1(z)$ that is greater than zero on $[0,\ell)$ will propagate with approximate speed $-\mu_1 \, p'(0^+)$. Noting that when $\mu_1 >\mu_2$ the jump condition~\eqref{eq:ppjump} gives $-\mu_1 \, p'(0^+) > -\mu_1 \, p'(0^-)$, we deduce that such a perturbation will separate from the rest of the travelling wave $\rho_1(z)$.
\end{remark}

\subsection{Quantitative comparison between individual-based and continuum models}
\label{sec:numres}

\subsubsection{One cell population}
We begin by comparing the computational simulation results of our individual-based model in the case of one cell population with the numerical solutions of equation~\eqref{eq:pde1}. For consistency with equation \eqref{eq:pde1}, we set $M=1$ and we drop the index $h=1$ both from the functions and from the parameters of the individual-based model. A complete description of the setup of numerical simulations is given in Appendix~\ref{appendix11} and Appendix~\ref{appendix21}. In particular, we use the following definition of the net growth rate $G$
\beq
\label{Gnum}
G(p) = \frac{1}{2 \, \pi} \, \arctan(\beta \, (P-p)) \quad \text{with} \quad \beta>0,
\eeq
so that assumptions~\eqref{as:G} are satisfied.

\paragraph{Travelling fronts.} Figure~\ref{Fig2}, along with the video accompanying it (\emph{vid.} Online~Resource~1), demonstrates that there is an excellent quantitative match between the numerical solutions of equation~\eqref{eq:pde1} and the computational simulation results of our individual-based model.
In agreement with the results established by Theorem~\ref{th:TW1}, the cell density is nonincreasing and connects the homogeneous steady state $\rho \equiv \Pi^{-1}(P)$ to the homogeneous steady state $\rho \equiv 0$. Accordingly, the cell pressure is nonincreasing and connects the homogeneous steady state $p \equiv P$ to the homogeneous steady state $p \equiv 0$. 
\begin{figure}
\centering
\includegraphics[width=\textwidth]{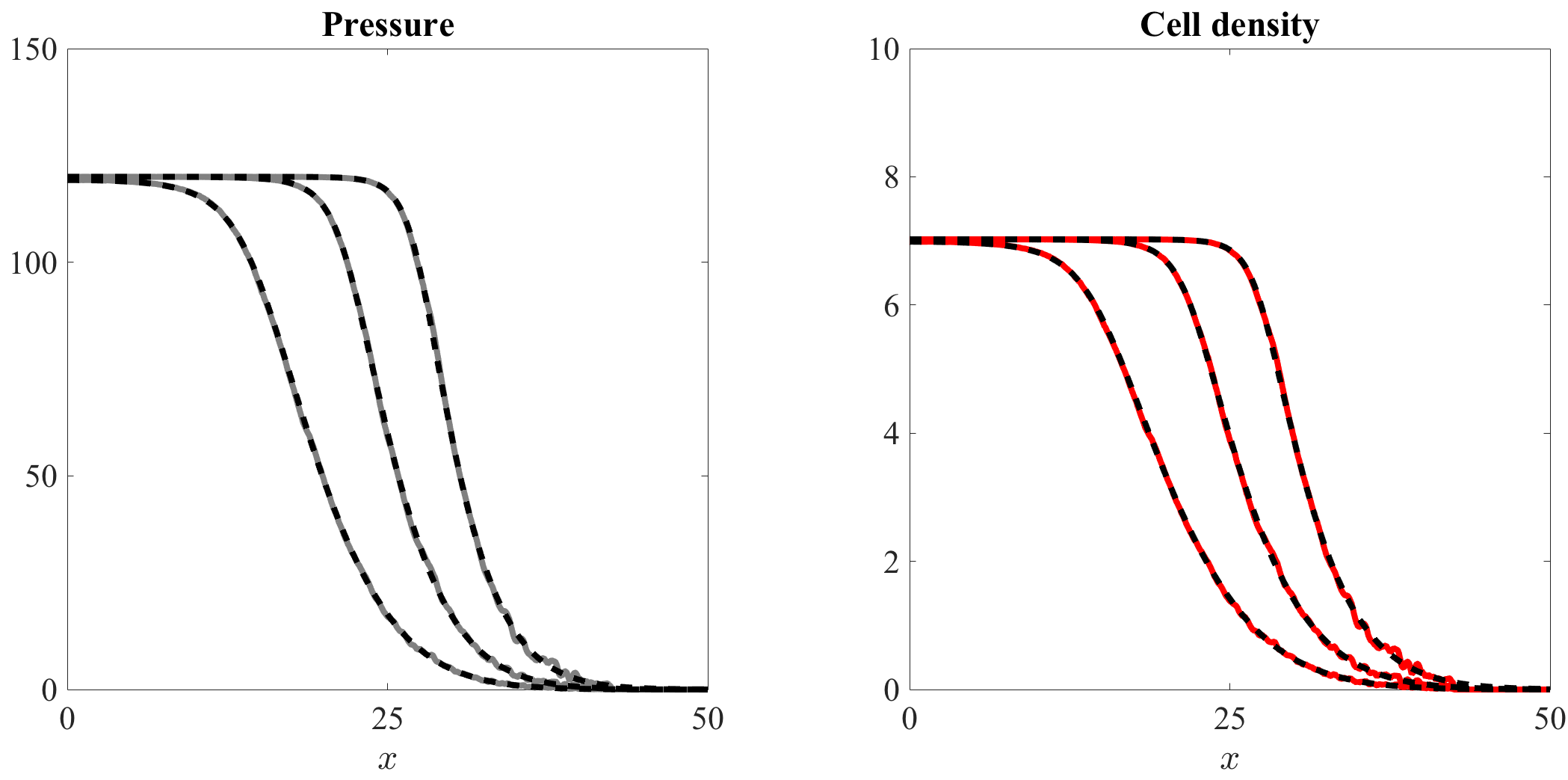}
\caption{{\bf Travelling fronts.} Comparison between the computational simulation results of our individual-based model in the case of one cell population (solid lines) and the numerical solutions of the continuum model~\eqref{eq:pde1} (dashed lines). The left and right panel display, respectively, the pressure and the cell density at three successive time instants, \emph{i.e.} $t=10$ (left curves), $t=15$ (middle curves) and $t=20$ (right curves). Values of the pressure and the cell density are in units of $10^4$. Simulations were carried out using a barotropic relation that satisfies assumptions~\eqref{as:p} and the definition~\eqref{Gnum} of $G(p)$, with the homeostatic pressure $P=120 \times 10^4$ and the coefficient $\beta=4 \times 10^{-6}$. A complete description of the numerical simulation setup is given in Appendix~\ref{appendix11} and Appendix~\ref{appendix21}. Sample dynamics of the pressure and the cell density obtained from the individual-based model and the continuum model are shown in the video accompanying this figure (\emph{vid.} Online~Resource~1)}
\label{Fig2}
\end{figure}

\paragraph{Higher values of $\beta$ lead to higher speed of invasion.} 
Figure~\ref{Fig3}, along with the video accompanying it (\emph{vid.} Online~Resource~2), indicates that, as one would expect, increasing the value of the parameter $\beta$ in the definition~\eqref{Gnum} of $G(p)$ accelerates the growth of the cell population, thus leading to a higher speed of invasion.
\begin{figure}
\centering
\includegraphics[width=\textwidth]{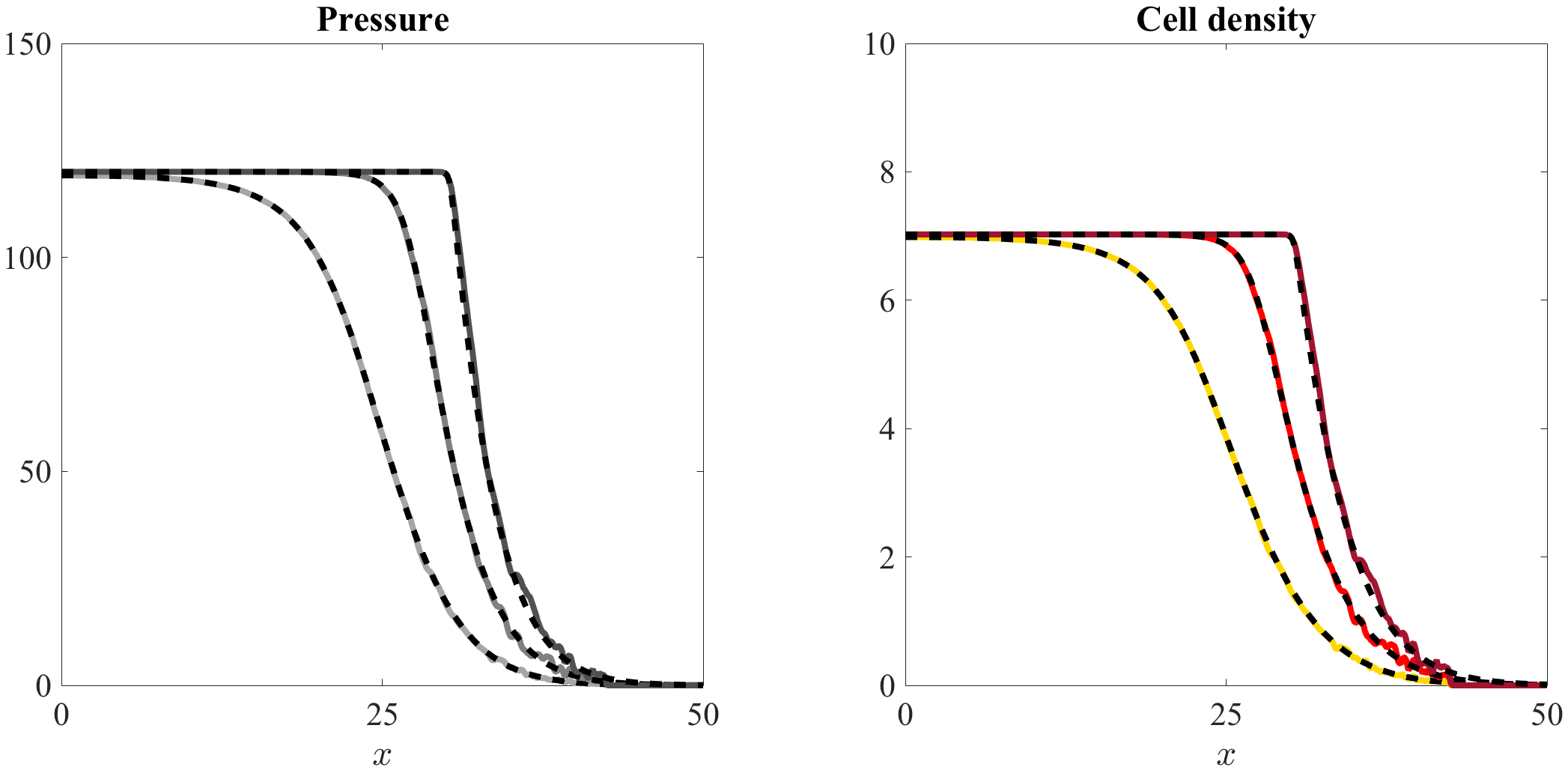}
\caption{{\bf Higher values of $\beta$ lead to higher speed of invasion.} Comparison between the computational simulation results of our individual-based model in the case of one cell population (solid lines) and the numerical solutions of the continuum model~\eqref{eq:pde1} (dashed lines). The left and right panel display, respectively, the pressure and the cell density at the time instant $t=15$ for increasing values of the parameter $\beta$ in the definition~\eqref{Gnum} of $G(p)$, \emph{i.e.} $\beta = 1.5 \times 10^{-6}$ (light grey and yellow lines), $\beta = 4 \times 10^{-6}$ (middle grey and red lines) and $\beta = 4 \times 10^{-5}$ (dark grey and brown lines). Values of the pressure and the cell density are in units of $10^4$. Simulations were carried out using a barotropic relation that satisfies assumptions~\eqref{as:p} with the homeostatic pressure $P=120 \times 10^4$. A complete description of the numerical simulation setup is given in Appendix~\ref{appendix11} and Appendix~\ref{appendix21}. Sample dynamics of the pressure and the cell density obtained from the individual-based model and the continuum model are shown in the video accompanying this figure (\emph{vid.} Online~Resource~2)}
\label{Fig3}
\end{figure}

\paragraph{Differences between the outcomes of individual-based and continuum models in the presence of sharp transitions from high to low cell densities.} The results presented so far indicate that there is an excellent agreement between the computational simulation results of our individual-based model and the solutions of the corresponding continuum models. However, due to extinction phenomena related to stochasticity effect that occur in the individual-based model for low cell densities, we expect differences between the outcomes of the two modelling approaches to emerge in the presence of sharp transitions from high to low total cell densities. In order to verify this hypothesis, exploiting the asymptotic results of Perthame {\it et al.} ~\cite{perthame2014hele} -- who have shown that under the barotropic relation~\eqref{def:p} sufficiently high values of the parameter $\gamma$ lead equation~\eqref{eq:pde1} to develop sharper invasion fronts -- we compare the computational simulation results of our individual-based model in the case of one cell population with the numerical solutions of equation~\eqref{eq:pde1} under the barotropic relation~\eqref{def:p} for increasing values of $\gamma$. A complete description of the setup of numerical simulations is given in Appendix~\ref{appendix11} and Appendix~\ref{appendix21}. The results obtained are summarised by Figure~\ref{Fig4} which shows that larger values of the parameter $\gamma$ can bring about sharper invasion fronts, which come along with more abrupt variations in the cell density, thus leading to more evident differences between the computational simulation results of the stochastic individual-based model and the numerical solutions of equation~\eqref{eq:pde1} at the front of invasion. Ultimately, this causes the invasion front of the individual-based model to travel at the same speed but behind the front of the corresponding continuum model.
\begin{figure}
\centering
\includegraphics[width=\textwidth]{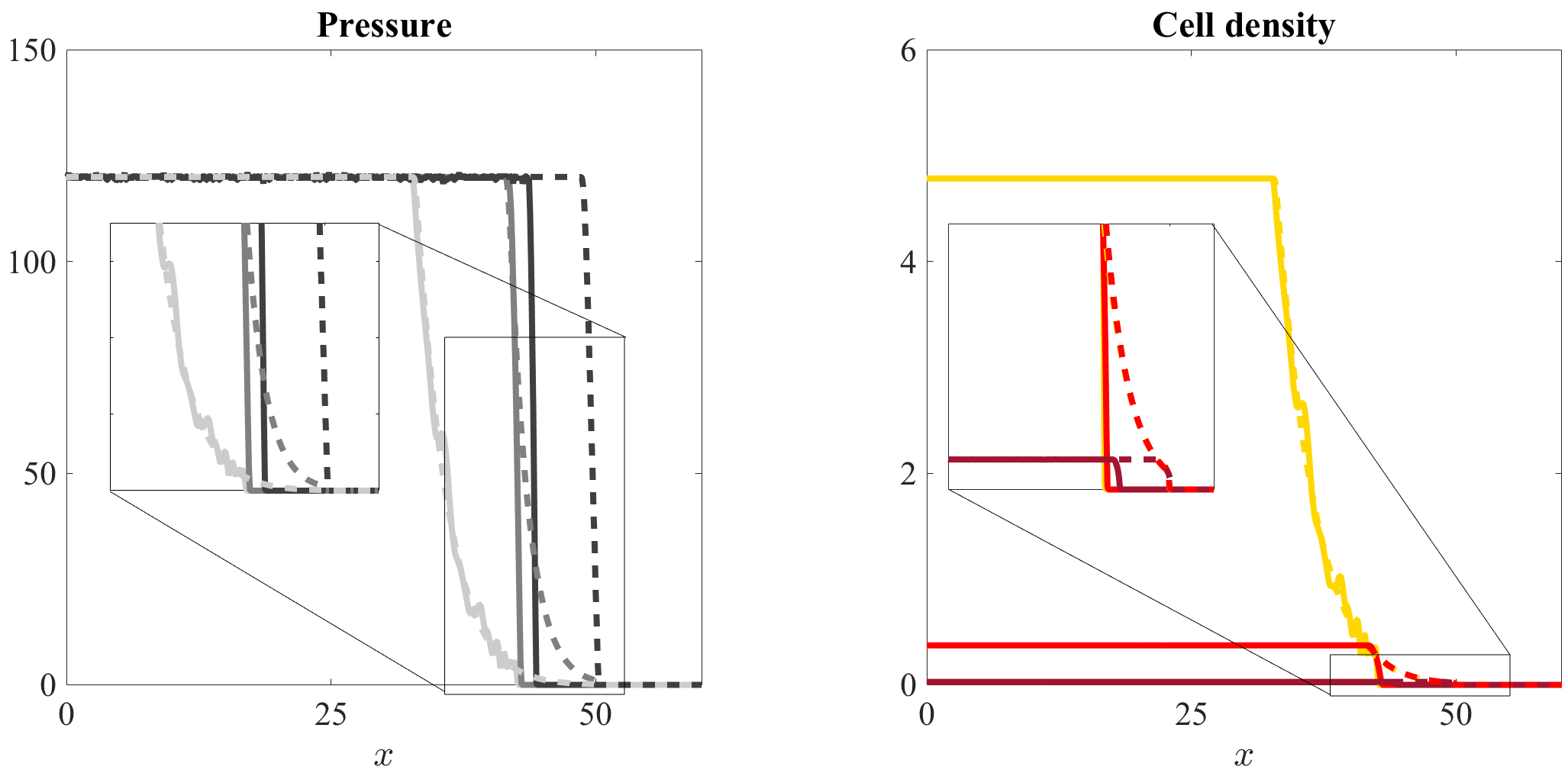}
\caption{{\bf Differences in the outcomes of individual-based and continuum models in the presence of sharp transitions from high to low cell densities.} Comparison between the computational simulation results of our individual-based model in the case of one cell population (solid lines) and the numerical solutions of the continuum model~\eqref{eq:pde1} (dashed lines). The left and right panel display, respectively, the pressure and the cell density at the time instant $t=15$ for increasing values of the parameter $\gamma$ in the barotropic relation~\eqref{def:p}, \emph{i.e.} $\gamma = 1.2$ (light grey and yellow lines), $\gamma = 1.5$ (middle grey and red lines) and $\gamma=2$ (dark grey and brown lines). Values of the pressure and the cell density are in units of $10^5$. Magnifications of the curves near the invasion fronts are shown in the insets. Simulations were carried out using the definition ~\eqref{Gnum} of $G(p)$, with the homeostatic pressure $P=120 \times 10^5$ and the coefficient $\beta=4 \times 10^{-5}$. A complete description of the numerical simulation setup is given in Appendix~\ref{appendix11} and Appendix~\ref{appendix21}}
\label{Fig4}
\end{figure}

\subsubsection{Two cell populations}
We now turn to the case of two cell populations and we compare computational simulation results of our individual-based model with numerical solutions of the system of equations~\eqref{eq:pde3}. For consistency with the system of equations~\eqref{eq:pde3}, we choose $M=2$, and we set $G_1 \equiv G$ and $G_2 \equiv 0$ in the individual-based model. Full details of the setup of numerical simulations can be found in Appendix~\ref{appendix12} and Appendix~\ref{appendix22} for the results reported in Figures \ref{Fig5} and \ref{Fig6}, and Appendix~\ref{appendix13} and Appendix~\ref{appendix23} for the results reported in Figures \ref{Fig7} and \ref{Fig8}. In particular, we use the definition~\eqref{Gnum} of the net growth rate $G(p)$ with the homeostatic pressure $P=10 \times 10^4$ and we let the population of nonproliferating cells (\emph{i.e.} population $2$) be ahead of the population of proliferating cells (\emph{i.e.} population $1$) at the initial time $t=0$.

Figures \ref{Fig5} and \ref{Fig6}, along with the videos accompanying them (\emph{vid.} Online~Resource~3 and Online~Resource~4), demonstrate that there is an excellent quantitative match between the numerical solutions of the system of equations~\eqref{eq:pde3} and the computational simulation results of our individual-based model, both in the case where $\mu_1 \leq \mu_2$ and when $\mu_1 > \mu_2$. Over time, the pressure $p$ converges to the homeostatic pressure $P$ while the cell density $\rho_1$ converges to the corresponding value $\Pi^{-1}(P)$.

\paragraph{Travelling fronts and spatial segregation between the two cell populations.} If $\mu_1 \leq \mu_2$ (\emph{i.e.} if $\nu_1 \leq \nu_2$), in agreement with the results established by Theorem~\ref{th:TW2}, spatial segregation occurs and the two cell populations remain separated by a sharp interface (\emph{vid.} Figure \ref{Fig5} and Online~Resource~3). The population of nonproliferating cells stays ahead of the population of  proliferating cells and, over the regions where they are greater than zero, the cell densities are nonincreasing. The pressure itself is continuous across the interface between the two cell populations, whereas its first derivative jumps from a lower negative value to a larger negative value, \emph{i.e.} the sign of the jump coincides with $\sgn(\mu_2 - \mu_1)$ (\emph{cf.} the jump condition~\eqref{eq:ppjump}). 
\begin{figure}
\centering
\includegraphics[width=\textwidth]{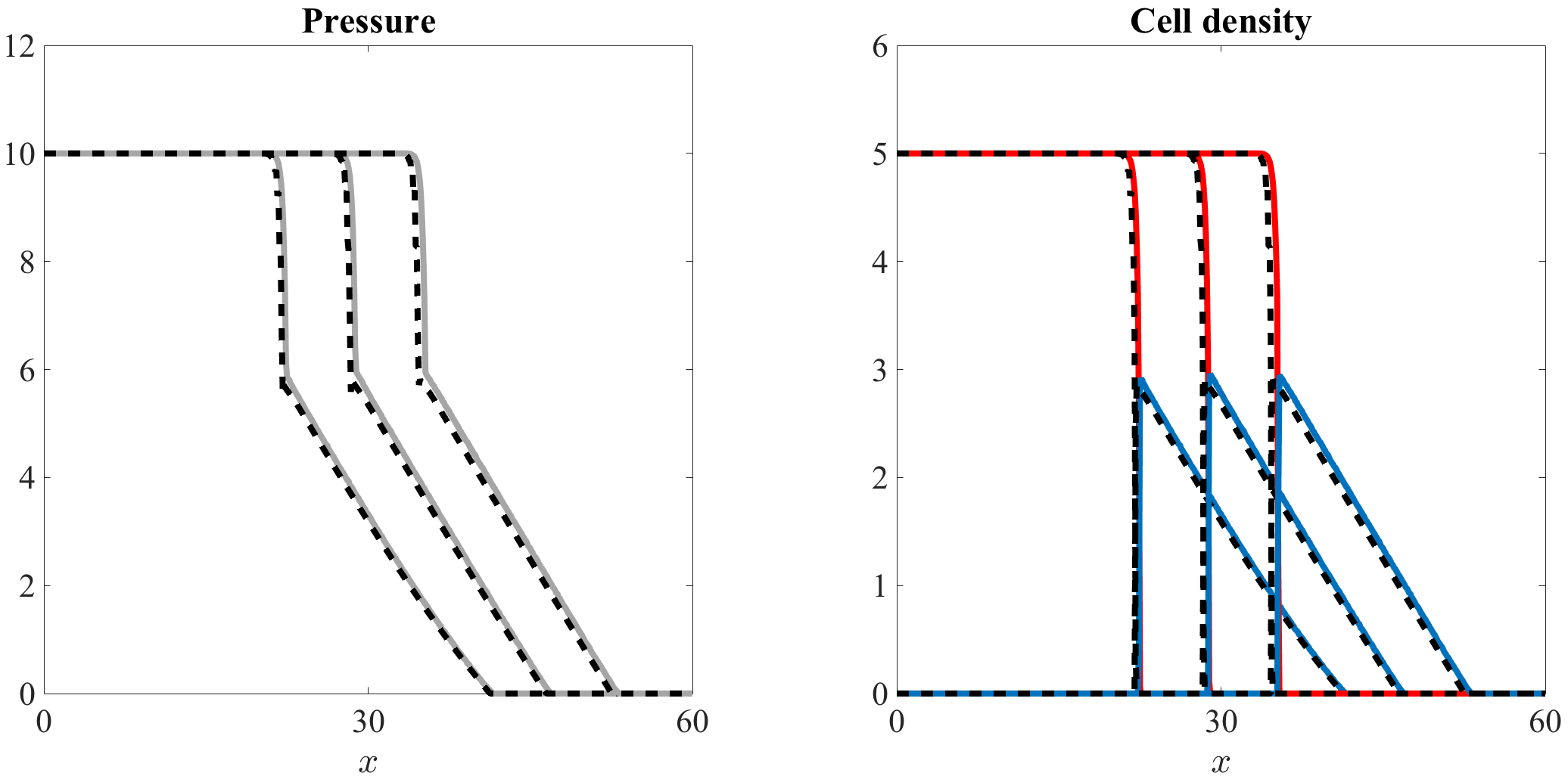}
\caption{{\bf Travelling fronts and spatial segregation between the two cell populations.} Comparison between the computational simulation results of our individual-based model in the case of two cell populations (solid lines) and the numerical solutions of the continuum model~\eqref{eq:pde3} (dashed lines), for $\mu_1 \leq \mu_2$ (\emph{i.e.} $\nu_1 \leq \nu_2$). The left and right panel display, respectively, the pressure and the cell densities of population 1 (red lines) and population 2 (blue lines) at three successive time instants, \emph{i.e.} $t=300$ (left curves), $t=450$ (middle curves) and $t=600$ (right curves). Values of the pressure and the cell densities are in units of $10^4$. Simulations were carried out using a barotropic relation that satisfies assumptions~\eqref{as:p} and the definition~\eqref{Gnum} of $G(p)$, with the homeostatic pressure $P=10 \times 10^4$ and the coefficient $\beta=4 \times 10^{-5}$. A complete description of the numerical simulation setup is given in Appendix~\ref{appendix12} and Appendix~\ref{appendix22}. Sample dynamics of the pressure and the cell density obtained from the individual-based model and the continuum model are shown in the video accompanying this figure (\emph{vid.} Online~Resource~3)}
\label{Fig5}
\end{figure}

\paragraph{Mixing between the two cell populations.} If $\mu_1 > \mu_2$ (\emph{i.e.} if $\nu_1 > \nu_2$) the cell population 2 is left behind by the cell population 1, which ultimately propagates alone (\emph{vid.} Figure \ref{Fig6} and Online~Resource~4). This is consistent with the heuristic argument provided in Remark~\ref{reminsta}, which suggests that the travelling-wave solutions of Theorem~\ref{th:TW2} are unstable in the case where $\mu_1 > \mu_2$. 
\begin{figure}
\centering
\includegraphics[width=\textwidth]{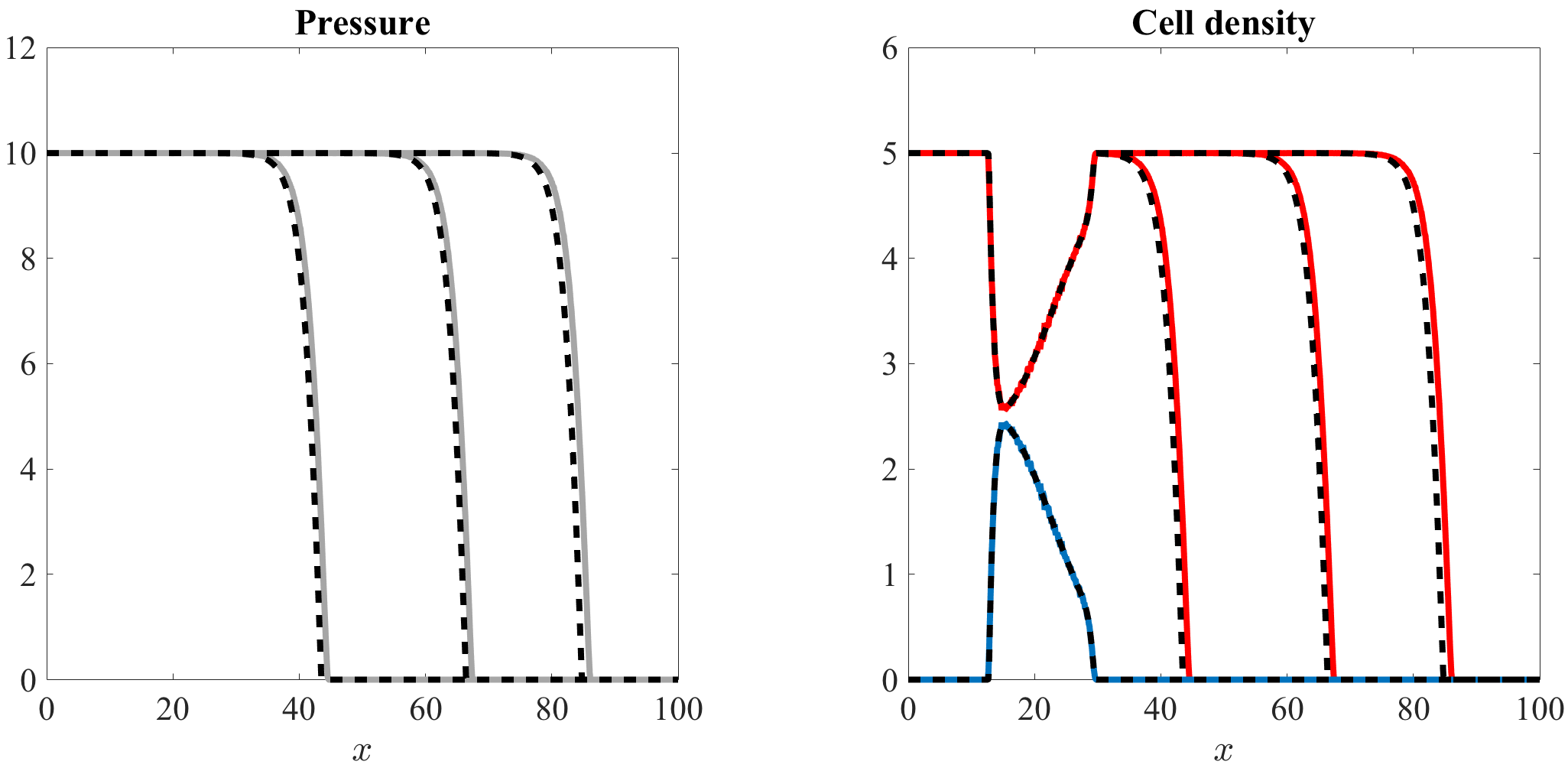}
\caption{{\bf Mixing between the two cell populations.} Comparison between the computational simulation results of our individual-based model in the case of two cell populations (solid lines) and the numerical solutions of the continuum model~\eqref{eq:pde3} (dashed lines), for $\mu_1 > \mu_2$ (\emph{i.e.} $\nu_1 > \nu_2$). The left and right panel display, respectively, the pressure and the cell densities of population 1 (red lines) and population 2 (blue lines) at three successive time instants, \emph{i.e.} $t=100$ (left curves), $t=150$ (middle curves) and $t=190$ (right curves). Values of the pressure and the cell densities are in units of $10^4$. Simulations were carried out using a barotropic relation that satisfies assumptions~\eqref{as:p} and the definition~\eqref{Gnum} of $G(p)$, with the homeostatic pressure $P=10 \times 10^4$ and the coefficient $\beta=4 \times 10^{-5}$. A complete description of the numerical simulation setup is given in Appendix~\ref{appendix12} and Appendix~\ref{appendix22}. Sample dynamics of the pressure and the cell density obtained from the individual-based model and the continuum model are shown in the video accompanying this figure (\emph{vid.} Online~Resource~4)}
\label{Fig6}
\end{figure}

\paragraph{Numerical simulations for more realistic barotropic relations.} In this paper, we have focused on barotropic relations that satisfy assumptions~\eqref{as:p}. However, as mentioned earlier in Section~\ref{intro}, more realistic barotropic relations satisfy the following conditions
\beq
\Pi(\rho) = 0 \; \text{ for } \; \rho \in [0,\rho^*) \quad \text{and} \quad \frac{\rmd \Pi}{\rmd \rho}>0 \; \text{ for } \; \rho>\rho^*,
\label{as:pmr}
\eeq
with $0<\rho^*<\Pi^{-1}(P)$. Figures~\ref{Fig7}~and~\ref{Fig8}, along with the videos accompanying them (\emph{vid.} Online~Resource~5 and Online~Resource~6), display the computational simulation results of our stochastic individual-based model and the numerical solutions of the system of equations~\eqref{eq:pde3} obtained under a barotropic relation that satisfies the more general assumptions~\eqref{as:pmr} (\emph{vid.} Appendix~\ref{appendix13} and Appendix~\ref{appendix23} for a complete description of the numerical simulation setup). These computational results and numerical solutions clearly share the same properties as those of Figures~\ref{Fig5}~and~\ref{Fig6}. This supports the conclusion that the essentials of the results obtained using barotropic relations of the form~\eqref{as:p} remain intact even under the more realistic assumptions~\eqref{as:pmr}.
\begin{figure}
\centering
\includegraphics[width=\textwidth]{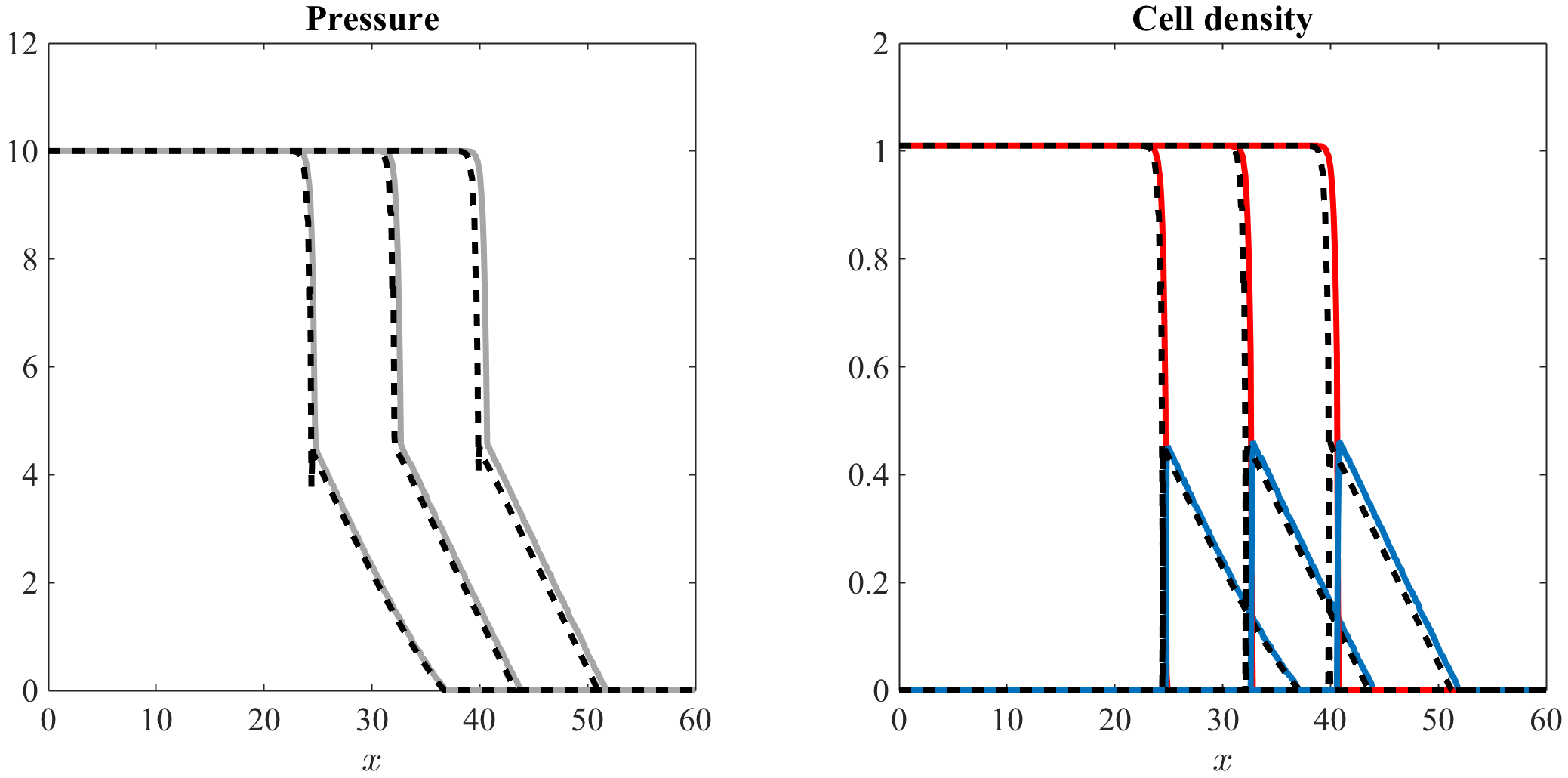}
\caption{{\bf Numerical simulations for more realistic barotropic relations.} Comparison between the computational simulation results of our individual-based model in the case of two cell populations (solid lines) and the numerical solutions of the continuum model~\eqref{eq:pde3} (dashed lines), for $\mu_1 \leq \mu_2$ (\emph{i.e.} $\nu_1 \leq \nu_2$). The left and right panel display, respectively, the pressure and the cell densities of population 1 (red lines) and population 2 (blue lines) at three successive time instants, \emph{i.e.} $t=300$ (left curves), $t=450$ (middle curves) and $t=600$ (right curves). Values of the pressure and the cell densities are in units of $10^4$. Simulations were carried out using a barotropic relation that satisfies the assumptions~\eqref{as:pmr} and the definition~\eqref{Gnum} of $G(p)$ with the homeostatic pressure $P=10 \times 10^4$ and the coefficient $\beta=4 \times 10^{-5}$. A complete description of the numerical simulation setup is given in Appendix~\ref{appendix13} and Appendix~\ref{appendix23}. Sample dynamics of the pressure and the cell density obtained from the individual-based model and the continuum model are shown in the video accompanying this figure (\emph{vid.} Online~Resource~5)}
\label{Fig7}
\end{figure}
\begin{figure}
\centering
\includegraphics[width=\textwidth]{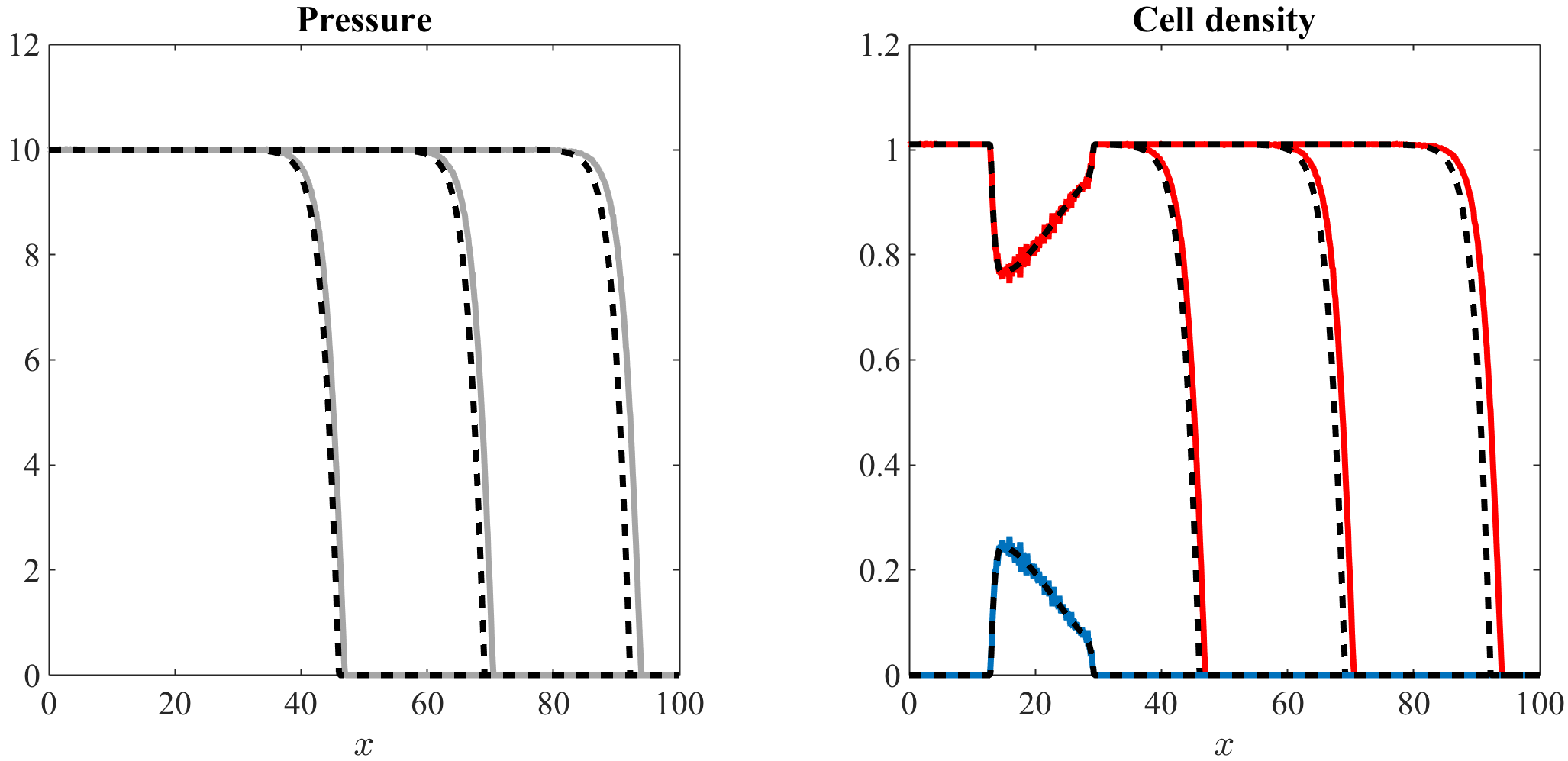}
\caption{{\bf Numerical simulations for more realistic barotropic relations.} Comparison between the computational simulation results of our individual-based model in the case of two cell populations (solid lines) and the numerical solutions of the continuum model~\eqref{eq:pde3} (dashed lines), for $\mu_1 > \mu_2$ (\emph{i.e.} $\nu_1 > \nu_2$). The left and right panel display, respectively, the pressure and the cell densities of population 1 (red lines) and population 2 (blue lines) at three successive time instants, \emph{i.e.} $t=100$ (left curves), $t=150$ (middle curves) and $t=200$ (right curves). Values of the pressure and the cell densities are in units of $10^4$. Simulations were carried out using a barotropic relation that satisfies assumptions~\eqref{as:pmr} and the definition~\eqref{Gnum} of $G(p)$ with the homeostatic pressure $P=10 \times 10^4$ and the coefficient $\beta=4 \times 10^{-5}$. A complete description of the numerical simulation setup is given in Appendix~\ref{appendix13} and Appendix~\ref{appendix23}. Sample dynamics of the pressure and the cell density obtained from the individual-based model and the continuum model are shown in the video accompanying this figure (\emph{vid.} Online~Resource~6)}
\label{Fig8}
\end{figure}

\section{Conclusions and research perspectives}
\label{sec:fin}
In summary, we have developed a simple, yet effective, stochastic individual-based model for the spatial dynamics of multicellular systems whereby cells undergo pressure-driven movement and pressure-dependent proliferation. We have shown that nonlinear partial differential equations commonly used to model the spatial dynamics of growing cell populations can be formally derived from the branching random walk that underlies our discrete model. Moreover, we have carried out a systematic comparison between the individual-based model and its continuum counterparts, both in the case of one single cell population and in the case of multiple cell populations with different biophysical properties. The outcomes of our comparative study demonstrate that the results of computational simulations of the individual-based model faithfully mirror the qualitative and quantitative properties of the solutions to the corresponding nonlinear partial differential equations. Ultimately, these results illustrate how the simple rules governing the dynamics of single cells in our individual-based model can lead to the emergence of complex spatial patterns of population growth observed in continuum models -- \emph{e.g.} travelling waves with composite shapes and sharp interfaces corresponding to spatial segregation between cell populations with different biophysical properties.

We have focussed on the case of one spatial dimension and considered barotropic relations that satisfy assumptions~\eqref{as:p}. However, the model presented here, as well as the related formal method to derive corresponding continuum models, can be adapted to higher spatial dimensions and more realistic barotropic relations of the form~\eqref{as:pmr}. In this regard, it would be interesting to use our stochastic individual-based model to further investigate the formation of finger-like patterns of invasion observed for the system of equations~\eqref{eq:pde3} posed on a two dimensional spatial domain~\cite{lorenzi2017interfaces}. Such spatial patterns resemble infiltrating patterns of cancer-cell invasion commonly observed in breast tumours~\cite{wang2012adipose}. An additional development of our study would be to compare the results presented here with those obtained from equivalent models defined on irregular lattices, as well as to investigate how our modelling approach could be related to off-lattice models of growing cell populations~\cite{drasdo2005coarse,motsch2018short,van2015simulating}. 

Our modelling framework can be easily extended to incorporate additional layers of biological complexity, such as nutrient-limited proliferation, undirected random cell movement, chemotaxis and haptotaxis. These are all lines of research that we will be pursuing in the near future.  


\appendix
\section{Details of numerical simulations of the individual-based model}
\label{appendix1}
We use a uniform discretisation of the interval $[0,100]$ that consists of 1001 points as the spatial domain (\emph{i.e.} the grid-step is $\chi=0.1$) and we choose the time-step $\tau=2\times10^{-3}$. We implement zero-flux boundary conditions by letting the attempted move of a cell be aborted if it requires moving out of the spatial domain. For all simulations, we use the definition~\eqref{Gnum} of $G(p)$ and we perform numerical computations in {\sc Matlab}. Further details specific either to the case of one single cell population or to the case of two cell populations are provided in the next subsections.

\subsection{Setup of numerical simulations for the case of one single cell population}
\label{appendix11}
For consistency with equation \eqref{eq:pde1}, we set $M=1$ and we drop the index $h=1$ both from the functions and from the parameters of the individual-based model. We define the rate $G$ in equations \eqref{e:div}-\eqref{e:qui} according to equation~\eqref{Gnum}. We set the homeostatic pressure $P=120 \times 10^4$ for the simulation results reported in Figure~\ref{Fig2} and Figure~\ref{Fig3}, while we choose $P=120 \times 10^5$ for the simulation results of Figure~\ref{Fig4}. Moreover, we choose $\beta = 4 \times 10^{-6}$ for the simulation results reported in Figure~\ref{Fig2}, $\beta = 4 \times 10^{-5}$ for the simulation results of Figure~\ref{Fig4}, and $\beta \in \left\{ 1.5 \times 10^{-6},  4 \times 10^{-6},  4 \times 10^{-5} \right\}$ for the simulation results of Figure~\ref{Fig3}. We set $\nu=0.02$ in equations~\eqref{e:left}-\eqref{e:stay} and we define the pressure $p^k_i$ according to the following barotropic relation
$$
\Pi(\rho^k_i) = K_{\gamma} \, (\rho^k_i)^{\gamma} \quad \text{with}  \quad K_{\gamma} = \frac{\gamma+1}{\gamma} \quad \text{and} \quad \gamma>1,
$$
which satisfies conditions~\eqref{as:p}. We let $\gamma \in \left\{1.2, 1.5, 2 \right\}$ for the simulation results of Figure~\ref{Fig4}, while we choose $\gamma = 1.2$ for the simulation results reported in Figure~\ref{Fig2} and Figure~\ref{Fig3}. We use the initial cell density
$$
\rho^0_{i} = A \, \exp{\left(- b \, x_i^{2}\right)} \quad \text{with} \quad A = 2\times 10^{4} \quad \mbox{and} \quad b=4\times 10^{-3}.
$$
The results presented in Figure~\ref{Fig2} and Figure~\ref{Fig3} correspond to the average over three simulations of our individual-based model, while the results in Figure~\ref{Fig4} correspond to one single simulation  when $\gamma=1.5$ or $\gamma=2$ and the average over two simulations when $\gamma=1.2$. 

\subsection{Setup of numerical simulations for the case of two cell populations -- Figure~\ref{Fig5} and Figure~\ref{Fig6}}
\label{appendix12}
For consistency with the system of equations \eqref{eq:pde3}, we choose $M=2$, and we set $G_1 \equiv G$ and $G_2 \equiv 0$ in equations \eqref{e:div}-\eqref{e:qui}, where $G$ is defined according to equation~\eqref{Gnum} with the homeostatic pressure $P=10 \times 10^4$ and the factor $\beta = 4 \times 10^{-5}$. We set $\nu_1 = 0.01$ and $\nu_2 = 0.5$ in equations~\eqref{e:left}-\eqref{e:stay} for the simulation results reported in Figure~\ref{Fig5}, while we consider $\nu_1 = 0.5$ and $\nu_2 = 0.01$ for the simulation results of Figure~\ref{Fig6}. We define the pressure $p^k_i$ according to the following simplified barotropic relation
$$
\Pi(\rho^k_i) = K \, \rho^k_i \quad \text{with}  \quad K = 2,
$$
which satisfies conditions~\eqref{as:p}. We make use of the initial cell densities 
\beq
\label{ic111}
\rho^0_{1 i}  = A_{1} \exp{\left(-b_{1} \, x_i^{2}\right)} \quad \mbox{and} \quad \rho^0_{2 i}=\left\{
\begin{array}{ll}
0 , \quad \text{for } x_i \leq 13,
\\\\
A_{2} \exp{\left(-b_{2} (x_i-14)^{2}\right)}, \quad \text{for } x_i \in (13,29),
\\\\
0 , \quad \text{for } x_i \geq 29,
\end{array}
\right.
\eeq
where
$$
A_{1}=1.25\times10^{4}, \quad b_{1}=0.06, \quad A_{2}=2.5\times10^{4} \quad \mbox{and} \quad b_{2}=6\times 10^{-3}.
$$
The results presented in Figure~\ref{Fig5} and Figure~\ref{Fig6} correspond to one single simulation of our individual-based model.

\subsection{Setup of numerical simulations for the case of two populations -- Figure~\ref{Fig7} and Figure~\ref{Fig8}}
\label{appendix13}
For consistency with the system of equations \eqref{eq:pde3}, we choose $M=2$, and we set $G_1 \equiv G$ and $G_2 \equiv 0$ in equations \eqref{e:div}-\eqref{e:qui}, where $G$ is defined according to equation~\eqref{Gnum} with the homeostatic pressure $P=10 \times 10^4$ and the factor $\beta = 4 \times 10^{-5}$. We set $\nu_1 = 0.01$ and $\nu_2 = 0.5$ in equations~\eqref{e:left}-\eqref{e:stay} for the simulation results reported in Figure~\ref{Fig7}, while we consider $\nu_1 = 0.5$ and $\nu_2 = 0.01$ for the simulation results of Figure~\ref{Fig8}. We define the pressure $p^k_i$ according to the following barotropic relation
$$
\Pi(\rho^k_i) =  q \, (\rho_i^k- \rho^*)_+ \quad \mbox{where} \quad q=10 \quad \mbox{and} \quad \rho^* = r \, P \; \text{ with } \; r=10^{-3},
$$
which satisfies conditions~\eqref{as:pmr}. We make use of the initial cell densities~\eqref{ic111} with
$$
A_{1}=12.5\times10^{4}, \quad b_{1}=0.06, \quad A_{2}=25\times10^{4} \quad \mbox{and} \quad b_{2}=6\times 10^{-3}.
$$
The results presented in Figure~\ref{Fig7} and Figure~\ref{Fig8} correspond to one single simulation of our individual-based model.

\section{Details of numerical simulations of the continuum models}
\label{appendix2}
We let $x \in [0,100]$ and we construct numerical solutions for equation~\eqref{eq:pde1} and for the system of equations~\eqref{eq:pde3} complemented with zero Neumann boundary conditions. We use a finite volume method based on a time-splitting between the conservative and nonconservative parts. For the conservative parts, transport terms are approximated through an upwind scheme whereby the cell edge states are calculated by means of a high-order extrapolation procedure~\cite{leveque2002finite}, while the forward Euler method is used to approximate the nonconservative parts. We consider a uniform discretisation of the interval $[0,100]$ that consists of 1001 points and we perform numerical computations in {\sc Matlab}. For all simulations, we use the definition~\eqref{Gnum} of $G(p)$. Further details specific either to the case of one cell population or to the case of two cell populations are provided in the next subsections.

\subsection{Setup of numerical simulations for equation~\eqref{eq:pde1}}
\label{appendix21}
The rate $G$ is defined according to equation \eqref{Gnum} with the homeostatic pressure $P=120 \times 10^4$ for the numerical solutions reported in Figure~\ref{Fig2} and Figure~\ref{Fig3}, while we choose $P=120 \times 10^5$ for the numerical solutions of Figure~\ref{Fig4}. Moreover, we choose $\beta = 4 \times 10^{-6}$ for the numerical solutions reported in Figure~\ref{Fig2}, $\beta = 4 \times 10^{-5}$ for the numerical solutions of Figure~\ref{Fig4}, and $\beta\in\left\{ 1.5 \times 10^{-6},  4 \times 10^{-6},  4 \times 10^{-5} \right\}$ for the numerical solutions reported in Figure~\ref{Fig3}. We define the pressure $p$ according to the following barotropic relation 
$$
\Pi(\rho) = K_{\gamma} \, \rho^{\gamma} \quad \text{with}  \quad K_{\gamma} = \frac{\gamma+1}{\gamma} \quad \text{and} \quad \gamma>1,
$$
which satisfies conditions~\eqref{as:p}. We let $\gamma \in \left\{1.2, 1.5, 2 \right\}$ for the numerical solutions of Figure~\ref{Fig4}, while we choose $\gamma = 1.2$ for the numerical solutions reported in Figure~\ref{Fig2} and Figure~\ref{Fig3}. Given the parameter values used for the individual-based model in the case of one single cell population, we choose the mobility $\mu = 4.166\times10^{-7}$ for the numerical solutions reported in Figure~\ref{Fig2} and Figure~\ref{Fig3}, while we set $\mu = 4.166\times10^{-8}$ for the numerical solutions of Figure~\ref{Fig4}. In this way, both values of $\mu$ satisfy condition~\eqref{asymat} for $h=1$. We impose the initial condition
$$
\rho(0,x) = A \, \exp{\left(- b \, x^{2}\right)} \quad \text{with} \quad A = 2\times 10^{4} \quad \mbox{and} \quad b=4\times 10^{-3}.
$$

\subsection{Setup of numerical simulations for the system of equations~\eqref{eq:pde3} -- Figure~\ref{Fig5} and Figure~\ref{Fig6}}
\label{appendix22}
The rate $G$ is defined according to equation \eqref{Gnum} with the homeostatic pressure $P=10 \times 10^4$ and $\beta = 4 \times 10^{-5}$. Given the parameter values used for the individual-based model in the case of two cell populations, we choose the mobilities $\mu_1 = 2.5\times 10^{-7}$ and $\mu_2 = 1.25\times 10^{-5}$ for the numerical solutions reported in Figure~\ref{Fig5}, and the mobilities $\mu_1 = 1.25\times 10^{-5}$ and $\mu_2 = 2.5\times 10^{-7}$ for the numerical solutions of Figure~\ref{Fig6}. This ensures that conditions~\eqref{asymat} for $h=1,2$ are satisfied. We define the pressure $p$ according to the following simplified barotropic relation 
$$
\Pi(\rho) = K \, \rho \quad \text{with}  \quad K = 2,
$$ 
which satisfies conditions~\eqref{as:p}. We impose the initial conditions
\beq
\label{ic111pde}
\rho^0_{1}(0,x)  = A_{1} \exp{\left(-b_{1} \, x^{2}\right)} \quad \mbox{and} \quad \rho^0_{2}=\left\{
\begin{array}{ll}
0 , \quad \text{for } x \leq 13,
\\\\
A_{2} \exp{\left(-b_{2} (x-14)^{2}\right)}, \quad \text{for } x \in (13,29),
\\\\
0 , \quad \text{for } x \geq 29,
\end{array}
\right.
\eeq
where
$$
A_{1}=1.25\times10^{4}, \quad b_{1}=0.06, \quad A_{2}=2.5\times10^{4} \quad \mbox{and} \quad b_{2}=6\times 10^{-3}.
$$

\subsection{Setup of numerical simulations for the system of equations~\eqref{eq:pde3} -- Figure~\ref{Fig7} and Figure~\ref{Fig8}}
\label{appendix23}
The rate $G$ is defined according to equation \eqref{Gnum} with the homeostatic pressure $P=10 \times 10^4$ and $\beta = 4 \times 10^{-5}$. Given the parameter values used for the individual-based model in the case of two cell populations, we choose the mobilities $\mu_1 = 2.5\times 10^{-7}$ and $\mu_2 = 1.25\times 10^{-5}$ for the numerical solutions reported in Figure~\ref{Fig7}, and the mobilities $\mu_1 = 1.25\times 10^{-5}$ and $\mu_2 = 2.5\times 10^{-7}$ for the numerical solutions of Figure~\ref{Fig8}. This ensures that conditions~\eqref{asymat} for $h=1,2$ are satisfied. We define the pressure $p$ according to the following barotropic relation 
$$
\Pi(\rho) =  q \, (\rho- \rho^*)_+ \quad \mbox{where} \quad q=10 \quad \mbox{and} \quad \rho^* = r \, P \; \text{ with } \; r=10^{-3},
$$ 
which satisfies conditions~\eqref{as:pmr}. We impose the initial conditions~\eqref{ic111pde} with
$$
A_{1}=12.5\times10^{4}, \quad b_{1}=0.06, \quad A_{2}=25\times10^{4} \quad \mbox{and} \quad b_{2}=6\times 10^{-3}.
$$

\begin{acknowledgements}
FRM is funded by the Engineering and Physical Sciences Research Council (EPSRC). TL and FRM gratefully acknowledge Dirk Drasdo and Lu\'is Neves de Almeida for insightful discussions.

\end{acknowledgements}

\bibliographystyle{spmpsci}      
\bibliography{Manuscript_Revision}   


\end{document}